\documentclass[11.5pt,oneside,reqno,english]{amsart}

\usepackage[T1]{fontenc}
\usepackage[latin9]{inputenc}
\usepackage{geometry}
\usepackage{babel}
\usepackage{amsbsy}
\usepackage{amstext}
\usepackage{amsthm}
\usepackage{amssymb}
\usepackage{braket}

\usepackage{amsaddr}

\usepackage{xcolor}

\usepackage{graphicx}

\usepackage[unicode=true,pdfusetitle,
 bookmarks=true,bookmarksnumbered=false,bookmarksopen=false,
 breaklinks=false,pdfborder={0 0 1},backref=false,colorlinks=false]
 {hyperref}

\makeatletter

\numberwithin{figure}{section}

\theoremstyle{plain}
\newtheorem{thm}{\protect\theoremname}[section]

\newtheorem{cor}[thm]{\protect\corollaryname}
\theoremstyle{definition}

\theoremstyle{definition}
\newtheorem{defn}[thm]{\protect\definitionname}
\newtheorem{example}[thm]{\protect\examplename}
\theoremstyle{plain}
\newtheorem{lem}[thm]{\protect\lemmaname}
\theoremstyle{definition}

\theoremstyle{remark}
\newtheorem{rem}[thm]{\protect\remarkname}

\newcommand{\norm}[1]{\left|\left| #1\right|\right|}

\newcommand{\cA}{\mathcal{A}}

\makeatother

\providecommand{\definitionname}{Definition}
\providecommand{\examplename}{Example}
\providecommand{\lemmaname}{Lemma}
\providecommand{\theoremname}{Theorem}
\providecommand{\conjecturename}{Conjecture}
\providecommand{\corollaryname}{Corollary}
\providecommand{\problemname}{Problem}
\providecommand{\remarkname}{Remark}
\providecommand{\propositionname}{Proposition}

\begin{document}

\title{Additive uncertainty relations via pseudospectral methods}

\author{Alexander Cerjan}
\email{awcerja@sandia.gov}
\address{Center for Integrated Nanotechnologies, Sandia National Laboratories, Albuquerque, New Mexico 87185, USA}

\author{Lara Ismert}
\email{ismert@unm.edu}
\address{Department of Mathematics and Statistics, University of New Mexico, Albuquerque, New Mexico 87131, USA}

\author{Terry A.\ Loring}
\email{tloring@unm.edu}
\address{Department of Mathematics and Statistics, University of New Mexico, Albuquerque, New Mexico 87131, USA}

\date{\today}


\begin{abstract}
Given Hermitian observable $A$, the variance of $A$ is given by the difference of expectations $\langle A^2\rangle - \langle A\rangle^2$, where the expectation of an observable implicitly assumes a choice of state for the underlying system. As a sub-case of a general lower bound for the sum of variances for finitely many arbitrary elements in a $C^*$-algebra, we prove that the sum of the variances for $d$ observables $A_1, \dots , A_d$ is greater than some expression in the expectations $\langle i[A_j,A_k]\rangle$ for $j\neq k$ plus the square of the smallest singular value of the associated localizer. In special cases, this expression is just a root mean square average. Our technique also derives a general containment between the Clifford pseudospectra and quadratic pseudospectra for $d$ observables.
\end{abstract}

\maketitle


\section{Introduction}

Uncertainty results have, in the past, been about lower bounds on
the product of the variances of two Hermitian observables $A$ and $B$. The simplest
of these \cite{Robertson1929uncertainty} being the lower bound
\begin{equation}
\Delta_{\mathbf{v}}^{2}(A)\Delta_{\mathbf{v}}^{2}(B)
\geq
\frac{1}{4}\bigl|\left\langle i\left[A,B\right]\right\rangle _{\mathbf{v}}\bigr|^{2}
\end{equation}
for a vector state $\mathbf{v}$. Progress by Maccone and Pati
\cite{Maccone_2014} and others \cite{chen2015sum_uncertainty,Lian24}
has lead to new lower bounds on uncertainty, where the bound is below the sum, not the product, of variances.  
One of the inequalities by Maccone and Pati stated for a pair $A$ and $B$ of Hermitian observables on a finite-dimensional system is
\begin{equation}
\Delta_{\mathbf{v}}^{2}(A)+\Delta_{\mathbf{v}}^{2}(B)
\geq
\left|\left\langle \left(A\pm iB\right)\mathbf{v},\mathbf{w}\right\rangle \right|^{2}\mp\left\langle i\left[A,B\right]\right\rangle _{\mathbf{v}}\label{eq:Maccone_Pati_ineq},
\end{equation}
where $\mathbf{w}$ must be selected as a unit vector orthogonal to
the system state $\mathbf{v}$.

When working with vector states, we denote by $\left\langle X\right\rangle _{\mathbf{v}}$ the expectation $\left\langle X\mathbf{v},\mathbf{v}\right\rangle $ of an observable $X$ represented concretely on a Hilbert space $\mathcal{H}$ with respect to a unit vector ${\bf v}\in \mathcal{H}$. We will also write
\begin{equation*}
\Delta_{\mathbf{v}}^{2}(X)
:=
\bigl\langle X^{2}\bigr\rangle_{\mathbf{v}}-\bigl\langle X\bigr\rangle_{\mathbf{v}}^{2}
\end{equation*}
to denote the variance of $X$ with respect to the state ${\bf v}$. Last, we will frequently write
\begin{equation*}
\hat{X}
:=
X-\left\langle X\right\rangle _{\mathbf{v}}I
\end{equation*}
so that $\Delta_{\mathbf{v}}^{2}(X)
=\bigl\langle\hat{X}^{2}\bigr\rangle_{\mathbf{v}}.$

The sign in (\ref{eq:Maccone_Pati_ineq}) should be chosen so that
$\mp\left\langle i\left[A,B\right]\mathbf{v},\mathbf{v}\right\rangle $
is positive. The accuracy of this lower bound depends greatly on the
choice of orthogonal state. Indeed for any $\mathbf{v}$, Maccone
and Pati note that
\begin{equation}
\sup_{\mathbf{w}\perp\mathbf{v}}\left|\left\langle \left(A\pm iB\right)\mathbf{v},\mathbf{w}\right\rangle \right|^{2}=\left\Vert \left(A_{\mathbf{v}}\pm iB_{\mathbf{v}}\right)\mathbf{v}\right\Vert ^{2}.
\end{equation}
Since they also showed
\begin{equation}
\Delta_{\mathbf{v}}^{2}(A)+\Delta_{\mathbf{v}}^{2}(B)=\bigl\Vert\bigl(\hat{A}\pm i\hat{B}\bigr)\mathbf{v}\bigr\Vert^{2}\mp\left\langle i\left[A,B\right]\right\rangle _{\mathbf{v}}\label{eq:Maccone_Pati_equality}
\end{equation}
this means that if one can find an appropriate $\mathbf{w}$, the estimate (\ref{eq:Maccone_Pati_ineq}) can be made very tight.

Starting from (\ref{eq:Maccone_Pati_equality}) one can find lower
bounds that do not require choosing an auxiliary vector. Simply the
fact that a norm will be non-negatative allows us to conclude that
\begin{equation*}
\Delta_{\mathbf{v}}^{2}(A)+\Delta_{\mathbf{v}}^{2}(B)
\geq
\left|\left\langle i\left[A,B\right]\right\rangle _{\mathbf{v}}\right|.
\end{equation*}
We will generalize this form of uncertainty relation to any number of observables on a possibly infinite-dimensional system and to arbitrary states, with the complexity
of the the lower bound increasing along with the number of commutators.
For just three Hermitian observables, the result is still elegant:
\begin{equation*}
\Delta^{2}(A)+\Delta^{2}(B)+\Delta^{2}(C)\geq\sqrt{\left|\left\langle i\left[A,B\right]\right\rangle \right|+\left|\left\langle i\left[A,C\right]\right\rangle \right|+\left|\left\langle i\left[B,C\right]\right\rangle \right|}.\label{eq:basic_three_fold_for_pure_state}
\end{equation*}
This inequality has been found before. Indeed there is a version of
this that works for three unbounded operators \cite{Lian24}.  The bounded case, and many related inequalities, can be found in \cite{dodonov2025uncertainty}.

We can do better if we use a non-trivial lower bound on the norm term.
To get this lower bound we will use the pseudospectrum \cite{Trefethen2005Pseudospectra}.
In the case of three of more observables we instead use the Clifford
pseudospectrum. The inequality we will apply is a special case of
the min-max theorem (see \cite[p.~75]{Bhatia1997MatrixAnalysis},
for example). For a matrix $X$, let $s_{\min}(X)$ denote the smallest
singular value of $X$. The inequality we will use is just
\begin{equation*}
\min_{\|\mathbf{v}\|=1}\left\Vert X\mathbf{v}\right\Vert =s_{\min}(X).
\end{equation*}
Since $s_{\min}\left(X^{*}\right)=s_{\min}\left(X\right)$ for any
matrix $X$, the two equalities
\begin{equation*}
\Delta_{\mathbf{v}}^{2}(A)+\Delta_{\mathbf{v}}^{2}(B)
=
\bigl\Vert\bigl(\hat{A}\pm i\hat{B}\bigr)\mathbf{v}\bigr\Vert^{2}+\bigl|\left\langle i\left[A,B\right]\right\rangle _{\mathbf{v}}\bigr|.
\end{equation*}
lead to just one inequality,
\begin{equation}
\Delta_{\mathbf{v}}^{2}(A)+\Delta_{\mathbf{v}}^{2}(B)
\geq
\bigl(s_{\min}\bigl(\hat{A}+i\hat{B}\bigr)\bigr)^{2}+\left|\left\langle i\left[A,B\right]\right\rangle _{\mathbf{v}}\right|.\label{eq:new_inequal_two_hermitian_case}
\end{equation}

The term $s_{\min}\bigl(\bigl(\hat{A}+i\hat{B}\bigr)\mathbf{v}\bigr)$
is zero exactly when $\left\langle A\right\rangle _{\mathbf{v}}+i\left\langle B\right\rangle _{\mathbf{v}}$
is in the spectrum of $A+iB$. If we want to find a vector state that
is well-localized in both $A$ and $B$ then we should only look at
vectors where $\left\langle A\right\rangle _{\mathbf{v}}+i\left\langle B\right\rangle _{\mathbf{v}}$
is at least close to being an eigenvalue of $A+iB$. That is, we want
to look for states where $\left\langle A\right\rangle _{\mathbf{v}}+i\left\langle B\right\rangle _{\mathbf{v}}$
is in the set
\begin{equation}
\Lambda_{\epsilon}(X)
=
\left\{ z\in\mathbb{C} \,|\, s_{\min}(X-zI)<\epsilon
\right\} \label{eq:classical_pseudospectrum}
\end{equation}
for $X=\left\langle A\right\rangle _{\mathbf{v}}+i\left\langle B\right\rangle _{\mathbf{v}}$
and $\epsilon$ small. 

The set in (\ref{eq:classical_pseudospectrum}) is what is classically
known as the pseudospectrum \cite{Trefethen2005Pseudospectra} of
the matrix $X$. Since we think of $A$ and $B$ as two separate inputs
to uncertainty inequalities, we prefer to work with the function of
two real variables. This, we call the \emph{Clifford gap} 
\begin{equation}
\mu_{(x,y)}^{\mathrm{C}}(A,B)=s_{\min}\left(\left(A-xI\right)+i\left(B-yI\right)\right).
\end{equation}
For a normal matrix $N$ the function $s_{\min}\left(N-zI\right)$
grows mostly linearly, as it just measures the distance from $z$
to the spectrum of $N$. For an almost normal matrix, say $X=A+iB$
where $A$ and $B$ are almost commuting Hermitian matrices, $s_{\min}\left(X-zI\right)$
will have growth that is almost linear near a point in the spectrum. For a precise result, see \cite[Theorem 1.3]{filonov2012polynomials}. 

We can see that we have been working with both the Laplace-type operator
$\hat{A}^{2}+\hat{B}^{2}$
and fragments of the Dirac-type operator 
\begin{equation*}
\hat{A}\otimes\sigma_{x}+\hat{B}\otimes\sigma_{y}=\left[\begin{array}{cc}
0 & \hat{A}-i\hat{B}\\
\hat{A}+i\hat{B} & 0
\end{array}\right].
\end{equation*}
The equations
\begin{equation*}
\bigl(\hat{A}\mp i\hat{B}\bigr)\bigl(\hat{A}\pm i\hat{B}\bigr)=\hat{A}^{2}+\hat{B}^{2}\pm i\bigl[\hat{A},\hat{B}\bigr]
\end{equation*}
 used by Maccone and Pati are fragments of the one equation
\begin{equation*}
\left[\begin{array}{cc}
0 & \hat{A}-i\hat{B}\\
\hat{A}+i\hat{B} & 0
\end{array}\right]^{2}=\left[\begin{array}{cc}
\hat{A}^{2}+\hat{B}^{2} & 0\\
0 & \hat{A}^{2}+\hat{B}^{2}
\end{array}\right]+\left[\begin{array}{cc}
i\left[A,B\right] & 0\\
0 & -i\left[A,B\right]
\end{array}\right].
\end{equation*}
That is, the square of the Laplace-type operator is almost a square root of the Dirac-type operator when $A$ and $B$ almost commute.

This is a special case of a relation between the spectral localizer and the quadratic composite operator (\cite{Vides23}). If we add in a third Hermitian matrix, this becomes
\begin{equation*}
\left[\begin{array}{cc}
\hat{C} & \hspace{-0.2cm} \hat{A}-i\hat{B}\\
\hat{A}+i\hat{B} & -\hat{C}
\end{array}\right]^{2}
=
(\hat{A}^{2}+\hat{B}^{2}+\hat{C}^{2})\otimes I_{2}+\left[\begin{array}{cc}
i\left[A,B\right] & \hspace{-0.4cm} i[B,C]+[A,C]\\
i[B,C]-[A,C] & -i\left[A,B\right]
\end{array}\right].
\end{equation*}
In Section~\ref{sec:results_on_uncertainty}, we will see this leads to the lower bound 
\begin{equation}
\Delta^{2}(A)+\Delta^{2}(B)+\Delta^{2}(C)
\geq
\mu_{\boldsymbol{\lambda}}^{\mathrm{C}}(A,B,C)^{2}+\sqrt{\left|\left\langle i\left[A,B\right]\right\rangle \right|+\left|\left\langle i\left[A,C\right]\right\rangle \right|+\left|\left\langle i\left[B,C\right]\right\rangle \right|}\label{eq:sum_of_three_with_PS_and_RMS_expectations}
\end{equation}
where
$\boldsymbol{\lambda}=(\left\langle A\right\rangle ,\left\langle B\right\rangle ,\left\langle C\right\rangle)$. 
Indeed, this works just as well with mixed states and is really a
result about elements in a $C^{*}$-algebra.

There are variations on our results that work for non-Hermitian observables.  The study of uncertainty for non-Hermitian observables, typically Hamiltonians, is an active area of research.  See, for example, \cite{song2025Uncertainty_nonhermitian}.

We
will see that for the basic observables in a model of a Chern insulator
(Hamiltonian plus two position operators) the estimate in (\ref{eq:sum_of_three_with_PS_and_RMS_expectations})
is very tight, essentially an equality if we probe at the Fermi level
near an edge of the system. 

Once we get to four observables, the term in the lower bound involving
the six commutators is not so simple anymore. However, our main result
has the following corollary. This deals with the common situation
where most of the observables are for position and so commute with
each other, while the last observable is the Hamiltonian of a system
that only almost commutes with the various position operators \cite{Cerjan2024classifying,stoiber2025spec_loc_strong_TI}.
The $d$-variable Clifford pseudospectrum $\left(\mu_{\boldsymbol{\lambda}}^{\textrm{C}}(A_{1},\dots,A_{d})\right)^{2}$
will be reviewed in Section~\ref{sec:preliminaries_on_composite_operators}.

\begin{cor}
Let $A_{1},\dots,A_{d}$ be
Hermitian matrices. If $A_{j}A_{k}=A_{k}A_{j}$ for all $1\leq j,k\leq d-1$, then 
\begin{equation}
\sum_{j=1}^{d}\Delta_{\mathbf{v}}^{2}(A_{j})
\geq
\left(\mu_{\boldsymbol{\lambda}}^{\textrm{C}}(A_{1},\dots,A_{n})\right)^{2}+\sqrt{\sum_{j=1}^{d-1}\left\langle i\left[A_{j},A_{d}\right]\right\rangle _{\mathbf{v}}^{2}},
\end{equation}
where $\lambda_{j}=\left\langle A_{j}\right\rangle _{\mathbf{v}}$.
\end{cor}

\section{States and C*-algebras}

Our results work with mixed states and are not limited to states on finite-dimensional Hilbert space. Indeed, the most natural setting for these results is to states on a unital $C^{*}$-algebra $\mathcal{A}$. 
A state on $\mathcal{A}$ is a map $\varphi:\mathcal{A}\rightarrow\mathbb{C}$ that is linear, sends $I$ to $1$ and is positive in the sense that $\varphi(A^{*}A)\geq0$ for all $A\in\mathcal{A}$. The set of all states on $\mathcal{A}$ will be denoted $\mathcal{S}(\mathcal{A})$. If $\mathcal{A}=\mathbf{M}_{n}(\mathbb{C})$ then a density matrix $\rho\in \mathbf{M}_n$ determines a state via $\varphi(A)=\mathrm{Tr}(\rho A)$ and a unit vector $\mathbf{v}\in \mathbb{C}^n$ determines a pure state via $\varphi(A)=\bigl\langle A\mathbf{v},\mathbf{v}\bigr\rangle$. In general, a state is a function that returns expectations and satisfies some basic axioms.

The next lemma about positive linear functionals (and thus states) on $C^{*}$-algebras is used throughout the paper. 
\begin{lem}\label{lem:Kadison_Scwharz_for_functionals} For any positive linear
functional $\psi$ on $\mathcal{A}$ and any element $A\in\mathcal{A}$, 
\begin{equation}
\psi(A^{*}A)\geq|\psi(A)|^{2}.
\end{equation}
\end{lem}

\begin{proof}
Since $\psi$ is a positive linear functional and $B^{*}B\geq0$ for
any $B\in\mathcal{A}$, we have $0\leq\psi(B^{*}B)$ for any $B\in\mathcal{A}$.
In particular, for $A\in\mathcal{A}$ and $\hat{A}:=A-\psi(A)$, we
find: 
\begin{align*}
0 
& \leq \psi(\hat{A}^*\hat{A})\\
& = \psi\left((A-\psi(A)\right)^{*}(A-\psi(A))\\
 & =\psi\left(A^{*}A-\psi(A)A^{*}-\overline{\psi(A)}A+\overline{\psi(A)}\psi(A)\right)\\
 & =\psi(A^{*}A)-\psi(A)\psi(A^{*})-\overline{\psi(A)}\psi(A)+|\psi(A)|^{2}\\
 & =\psi(A^{*}A)-|\psi(A)|^{2},
\end{align*}
 where the last equality follows from $\psi(A^{*})=\overline{\psi(A)}$.
Therefore, $\psi(A^{*}A)\geq|\psi(A)|^{2}$. 
\end{proof}

One difficulty when dealing with an infinite-dimensional Hilbert space $\mathbb{H}$ is that now singular values are not defined. We can define $s_{\min}$ of a bounded linear operator $T$ as the conorm 
\begin{equation}
s_{\min}\left(T\right)=\min_{\|\mathbf{v}\|=1}\left\Vert T\mathbf{v}\right\Vert,
\end{equation}
but we no longer have the property $s_{\min}\left(T^{*}\right)=s_{\min}\left(T\right)$.
We will occasionally use this concept, as in  \cite{cerjan2025multivariable_PS}, but mostly we will only apply this notion to Hermitian
element where it is simply the distance from zero to the spectrum. As we have in mind physics applications, we will call this simpler notion the spectral gap, but it could just as well be called the minimal
spectral radius. To avoid confusion with singular values, we will
denote this as $\mu(A)$.

Suppose $A$ is a Hermitian element of the unital $C^{*}$-algebra
$\mathcal{A}$. The \emph{spectral gap} $\mu(A)$ is defined as
\begin{equation*}
\mu(A)
=
\min_{z\in\sigma(A)}|z|.
\end{equation*}
Recall that the norm of a Hermitian $A$ will equal the largest absolute
value of elements of its spectrum. Also, $A$ is deemed positive if
it is Hermitian with non-negative spectrum, or equivalently if $A=B^{*}B$
for some $B$ in the $C^{*}$-algebra.  

\begin{lem}
\label{lem:bounds_on_state_values} Given a positive element $A$ of
a unital $C^{*}$-algebra $\mathcal{A}$ and a state $\varphi\in\mathcal{S}(\mathcal{A})$, we have
\begin{equation*}
\mu(A)\leq\varphi(A)\leq\|A\|.
\end{equation*}
Moreover there are states $\varphi_{0}$ and $\varphi_{1}$ so that
$\varphi_{0}(A)=\mu(A)$ and $\varphi_{1}(A)=\|A\|$.
\end{lem}

\begin{proof}
This first fact follows from the order relation $\mu(A)I\leq A\leq\|A\|$.
To show the existence of states we can first prove this in the case
$\mathcal{A}=C[a,b]$ and $A(x)=x$, and then extend via the Hahn-Banach
theorem. By the Gelfand-Naimark Theorem, the result thus holds for an arbitrary Hermitian element in $\mathcal{A}$.
\end{proof}
We will need some basic facts about tensor products of $C^{*}$-algebras,
just in the context of 
\begin{equation*}
\mathcal{A}\otimes\mathbf{M}_{N}(\mathbb{C})\cong\mathbf{M}_{N}(\mathcal{A}).
\end{equation*}
The $C^{*}$-algebra $\mathbf{M}_{N}:=\mathbf{M}_{N}(\mathbb{C})$ is nice enough that we need not worry much about the norm on $\mathcal{A}\otimes\mathbf{M}_{N}$.
If $\mathcal{A}$ is represented on a Hilbert space $\mathbb{H}$, we can use the norm on
$\mathbf{M}_{N}(\mathcal{A})$ arising from its action on $\mathbb{H}\otimes\mathbb{C}^{N}$.
Given a state $\psi$ on $\mathcal{A}$ and a state $\varphi$ on
$\mathbf{M}_{N}$, the {\em product state} $\psi\otimes\varphi$
is defined by 
\begin{equation*}
\left(\psi\otimes\varphi\right)(A\otimes M)=\psi(A)\varphi(M)
\end{equation*}
and extends continuously and linearly to the full tensor product. There are more complex (entangled) states on $\mathcal{A}\otimes\mathbf{M}_{N}$,
but product states are all we need for the present paper.

Fix a state $\psi$ on $\mathcal{A}$. Given $A\in \mathcal{A}$, we define
\begin{equation}
\hat{A}=A_{\psi}=A-\psi(A)I
\end{equation}
so that
\begin{equation}
\Delta_{\psi}^{2}(A)=\psi(\hat{A}^{*}\hat{A}).
\end{equation}
See \cite{magajna2016variance_of_operators} for details. When $A$ is Hermitian, $\Delta_\psi^2(A) = \psi(\hat{A}^2).$ In the next section, we also need the following generalization of a commutator for Hermitian observables. Instead of the usual $\left[A,B\right]=AB-BA$ for Hermitian $A,B\in \cA$, we will have
\begin{equation}
\left[A,B\right]_{{\textstyle *}}=A^{*}B-B^{*}A.
\end{equation}
A reasonable name for this binary operator is $*$-commutator, but there does not seem to
be a standard name for this. As in the Hermitian case, the $*$-commutator of two arbitrary elements $A,B\in \cA$ will always be skew-Hermitian, meaning $i[A,B]_*$ is Hermitian.
The main distinction with the Hermitian case is that, for arbitrary $\lambda\in \mathbb{C}$,
\begin{equation*}
\left[\lambda I,B\right]_{{\textstyle *}}
\neq
0,
\end{equation*}
so in particular, we in general have
\begin{equation*}
[\hat{A},B]_{\textstyle *} \neq [A,B]_{\textstyle *},
\end{equation*} 
in contrast to the Hermitian setting where we enjoy the fact that $[\hat{A},B] = [A,B]$ for all Hermitian $A,B\in \cA.$
This comes into play in Definition\autoref{def:R_psi-matrix}.

\section{Multivariable pseudospectra}\label{sec:preliminaries_on_composite_operators}

The spectral localizer is a Dirac-type operator formed out of $d$ Hermitian elements $A_{1},\dots,A_{d}$ of $\mathcal{A}$.
Select $\Gamma_{1},\dots,\Gamma_{d}\in \mathbf{M}_N$ so that 
\begin{equation}
\Gamma_{j}^{*}=\Gamma_{j},\quad(1\leq j\leq d),
\end{equation}
\begin{equation}
\Gamma_{j}^{2}=I_{N},\quad(1\leq j\leq d),
\end{equation}
\begin{equation}
\Gamma_{j}\Gamma_{k}=-\Gamma_{k}\Gamma_{j},\quad(1\leq j,k\leq d, j\neq k).
\end{equation}
Usually we select this with  $N=2^{\left\lfloor (d-1)/2\right\rfloor}$, which is the smallest possible choice, corresponding to an irreducible representation of the associated complex Clifford algebra. In this paper, we extend the notion of the spectral localizer to encompass a $d$-tuple $(A_1,\dots, A_d)$ of arbitrary elements $A_1,\dots, A_d\in \mathcal{A}$.

\begin{defn}\label{def:spectral_localizer} 
Given a $d$-tuple $\mathbb{A}=(A_{1},\dots,A_{d})$ of elements $A_{j}$ of $\mathcal{A}$ and $\vec{\lambda}=(\lambda_{1},\dots,\lambda_{d})\in \mathbb{C}^d$, the {\em spectral localizer} $L_{\vec{\lambda}}(\mathbb{A})\in\mathcal{A}\otimes\mathbf{M}_{N}$
is defined as
\begin{equation}
L_{\vec{\lambda}}(\mathbb{A})
:=
\sum_{j=1}^{d}(A_{j}-\lambda_{j}I)\otimes\Gamma_{j}.
\end{equation}
\end{defn}

Notice we do not require the $\lambda_{j}$ to be real, nor the
$A_{j}$ to be Hermitian, or even normal. However, when calculating $K$-theory values \cite{Loring15} we need to make all these restrictions, or at least most of them \cite{cerjan2023Line_gapped_spec_loc}. 

A useful convention is to write $L(\mathbb{A})$ to mean $L_{\vec{0}}(\mathbb{A})$.
Since 
\begin{equation*}
L_{\vec{\lambda}}(A_{1},\dots,A_{d})=L(A_{1}-\lambda_{1}I,\dots,A_{d}-\lambda_{d}I)
\end{equation*}
we have, in particular, for any $\psi\in \mathcal{S}(\mathbb{A})$, 
\begin{equation*}
L_{\psi(\mathbb{A})}(A_{1},\dots,A_{d})=L(\hat{A}_{1},\dots,\hat{A}_{d});
\end{equation*}
where we set $\psi(\mathbb{A}):=\left(\psi(A_{1}),\dots\psi(A_{d})\right).$

\begin{defn}
We define the \emph{Clifford spectrum} $\Lambda(\mathbb{A})$ of $\mathbb{A}$,
a $d$-tuple of elements of $\mathcal{A}$, to be the set
\begin{equation*}
\Lambda(\mathbb{A})
:=
\left\{ \vec{\lambda}\in\mathbb{C}^{d}\mid L_{\vec{\lambda}}(\mathbb{A})\mbox{ is non-invertible}\right\} .
\end{equation*}
If we require all the $A_j$ to be Hermitian then we define the Clifford spectrum to be
\begin{equation*}
\Lambda(\mathbb{A})
:=
\left\{ \vec{\lambda}\in\mathbb{R}^{d}\mid L_{\vec{\lambda}}(\mathbb{A})\mbox{ is non-invertible}\right\} .
\end{equation*}
\end{defn}

\begin{rem}
There are lots of variations on how to form the spectral localizer
in the non-Hermitian case \cite{garcia2026clifford_Quadratic,cerjan2023Line_gapped_spec_loc},
and so there will be competing notions of Clifford spectrum. The same will be true of how we might form the quadratic composite operator and the various forms of pseudospectra. Our definitions are natural in the uncertainty-bound context but would be less useful for the study of lossy topological insulators than the approach used by Garcia \cite{garcia2026clifford_Quadratic}.
\end{rem}

\begin{rem}
To avoid confusion, we should give separate names to the two notions of Clifford spectrum that can apply when the $A_j$ are Hermitian. Mostly this should be clear from context.  When working in the context of just Hermitian $C^*$-algebraic variables the Clifford spectrum should be a subset of $\mathbb{R}^n$. In the context of random unitary observables, one might wish to limit the pool of possible elements for joint spectra to $d$-tuples of unit scalars. Our main results will be stated in the context of general observables and complex $\lambda_j$. All our special cases will be for Hermitian matrix observables and thus the $\lambda_j$ are restricted to being real.
\end{rem}

\begin{example}
A key example is when $\mathcal{A}=C(\Omega)$ for some compact
subset $\Omega$ of $\mathbb{R}^{d}$. Let $A_{j}$ be the $j$-th
coordinate function $A_j(\vec{x}):=x_j$ for all $\vec{x}=(x_1,\dots, x_d)\in \mathbb{R}^d$. These coordinate functions are real and therefore Hermitian. For each $\vec{\lambda}\in \mathbb{R}^d$, we can witness the localizer as a function $L_{\vec{\lambda}}(\mathbb{A}):\Omega\to M_N$ defined on a vector $\vec{x}\in\Omega$ by
\begin{equation}
\left[L_{\vec{\lambda}}(\mathbb{A})\right](\vec{x})
=
\sum_{j=1}^{d}(x_{j}-\lambda_{j})\otimes\Gamma_{j}.
\end{equation}
At ${\vec{x}} = {\vec{\lambda}}$, the localizer $L_{\vec{\lambda}}(\mathbb{A})$ is singular. 

As the $A_{j}$ are all Hermitian, so is $L_{\vec{\lambda}}(\mathbb{A})$. Note that a Hermitian operator is singular if and only if its square is singular. Observe
\begin{align*}
\bigl[L_{\vec{\lambda}}(\mathbb{A})\bigr]^{2}(\vec{x}) & =\left(\sum_{j=1}^{d}(x_{j}-\lambda_{j})\otimes\Gamma_{j}\right)^{2},
\end{align*}
simplifies to the sum of two types of terms. Cross-terms (with $j\neq k$) come in cancelling pairs:
\begin{equation*}
(x_{j}-\lambda_{j})(x_{k}-\lambda_{k})\otimes\Gamma_j\Gamma_{k}+(x_k-\lambda_k)(x_j-\lambda_j)\otimes \Gamma_k\Gamma_j
=
0,
\end{equation*}
since $\Gamma_{j}\Gamma_{k}=-\Gamma_{k}\Gamma_{j}$. Thus, the remaining terms are
\begin{equation*}
\left[L_{\vec{\lambda}}(\mathbb{A})\right]^{2}(\vec{x})=\sum_{j=1}^{d}(x_{j}-\lambda_{j})^{2}\otimes I,
\end{equation*}
so $L_{\vec{\lambda}}(\mathbb{A})$ is singular if and only if $\vec{\lambda}\in\Omega$. Thus the Clifford spectrum of the coordinate functions $A_1,\dots, A_d\in C(\Omega)$ is precisely $\Omega$. Altogether, this example illustrates the classical, commuting limit of a collection of observables, in which all coordinates can be simultaneously specified.
\end{example}

\begin{defn}
Given an $d$-tuple $\mathbb{A}=(A_{1},\dots,A_{d})$ of elements
$A_{j}$ of $\mathcal{A}$ and scalars $\vec{\lambda}=(\lambda_{1},\dots,\lambda_{d})$,
the {\em quadratic composite operator} $Q_{\vec{\lambda}}(\mathbb{A})\in\mathcal{A}$
is defined as
\begin{equation}
Q_{\vec{\lambda}}(\mathbb{A})
:=
\sum_{j=1}^{d}(A_{j}-\lambda_{j}I)^{*}(A_{j}-\lambda_{j}I).
\end{equation}
\end{defn}

We can try to define a joint spectrum of $A_1,\dots, A_d$ using the spectrum of $Q_{\vec{\lambda}}(\mathbb{A})$,
but it turns out that this is not very interesting. Indeed, the trouble is that even in the case of Hermitian elements, the set
\begin{equation}
\left\{ \vec{\lambda}\in\mathbb{R}^{d}\mid Q_{\vec{\lambda}}(\mathbb{A})\mbox{ is non-invertible}\right\}
\end{equation}
is too-often empty. On the other hand, we remain convinced that the Clifford spectrum is always non-empty \cite{cerjan2025multivariable_PS}. However, the the spectral gap of $Q_{\vec{\lambda}}(\mathbb{A})$
is closely related to the sum of uncertainty. This points to the fact that we need to look at pseudospectra when working in the noncommutative setting, not just spectra.
\begin{defn}
\label{def:quadratic_gap}
Given a $d$-tuple $\mathbb{A}:=(A_{1},\dots,A_{d})$ of elements in $\mathcal{A}$, the \emph{quadratic
gap of $\mathbb{A}$ at $\vec{\lambda}\in\mathbb{C}^{d}$} is 
\begin{equation*}
\mu_{\vec{\lambda}}^{\mathrm{Q}}(\mathbb{A})
:=
\sqrt{(Q_{\vec{\lambda}}(\mathbb{A}))},
\end{equation*}
and the \emph{Clifford gap  of $\mathbb{A}$ at $\vec{\lambda}\in\mathbb{C}^{d}$} is
\begin{equation*}
\mu_{\vec{\lambda}}^{\mathrm{C}}(\mathbb{A}):=\sqrt{\mu\left(\left(L_{\vec{\lambda}}(\mathbb{A})\right)^{*}L_{\vec{\lambda}}(\mathbb{A}))\right)}.
\end{equation*}
If all the $A_{j}$ are Hermitian, then the spectral localizer
is Hermitian and 
\begin{equation*}
\mu_{\vec{\lambda}}^{\mathrm{C}}(\mathbb{A})
=
\mu(L_{\vec{\lambda}}(\mathbb{A})).
\end{equation*}
\end{defn}

Associated to the gap functions $\vec{\lambda}\mapsto \mu^\mathrm{Q}_{\vec{\lambda}}(\mathbb{A})$ and $\vec{\lambda}\mapsto \mu^\mathrm{C}_{\vec{\lambda}}(\mathbb{A})$ are the $\epsilon$-pseudospectra of $\mathbb{A}$.
The {\em Clifford $\epsilon$-pseudospectrum} of $\mathbb{A}$ is 
\begin{equation*}
\Lambda_{\epsilon}^{\mathrm{C}}(\mathbb{A})=\left\{ \vec{\lambda}\in\mathbb{C}^{d}\mid\mu_{\vec{\lambda}}^{\mathrm{C}}(\mathbb{A})\leq\epsilon\right\} 
\end{equation*}
and the {\em quadratic $\epsilon$-pseudospectrum} of $\mathbb{A}$ is
\begin{equation*}
\Lambda_{\epsilon}^{\mathrm{Q}}(\mathbb{A})=\left\{ \vec{\lambda}\in\mathbb{C}^{d}\mid\mu_{\vec{\lambda}}^{\mathrm{Q}}(\mathbb{A})\leq\epsilon\right\} .
\end{equation*}

\begin{lem}
\label{lem:gaps_are_continuous}Let $\mathcal{A}$ be a unital $C^{*}$-algebra, let $\mathbb{A}:=(A_{1},\dots,A_{d})$ be a $d$-tuple of elements from $\mathcal{A}$. The functions $\vec{\lambda}\mapsto\mu_{\vec{\lambda}}^{\mathrm{C}}(\mathbb{A})$
and $\vec{\lambda}\mapsto\mu_{\vec{\lambda}}^{\mathrm{Q}}(\mathbb{A})$
are continuous.
\end{lem}

\begin{proof}
The function 
\begin{equation*}
\vec{\lambda}\mapsto\left(L_{\vec{\lambda}}(\mathbb{A})\right)^{*}L_{\vec{\lambda}}(\mathbb{A}),
\end{equation*}
for a fixed $\mathbb{A}$, is certainly continuous for the operator
norm. Since $M\mapsto\mu(M)$ is continuous on the set of Hermitian
elements, the Clifford gap must be continuous. The proof in the quadratic
case is essentially the same.
\end{proof}

\begin{thm}
Let $\mathcal{A}$ be a unital $C^{*}$-algebra, let $\mathbb{A}:=(A_{1},\dots,A_{d})$
be a $d$-tuple of elements from $\mathcal{A}$. For any $\epsilon\geq0$
\begin{equation*}
\Lambda_{\epsilon}^{\mathrm{C}}(\mathbb{A})=\left\{ \vec{\lambda}\in\mathbb{C}^{d}\mid\mu_{\vec{\lambda}}^{\mathrm{C}}(\mathbb{A})\leq\epsilon\right\} 
\end{equation*}
and
\begin{equation*}
\Lambda_{\epsilon}^{\mathrm{Q}}(\mathbb{A})=\left\{ \vec{\lambda}\in\mathbb{C}^{d}\mid\mu_{\vec{\lambda}}^{\mathrm{Q}}(\mathbb{A})\leq\epsilon\right\} 
\end{equation*}
are compact sets.
\end{thm}

\begin{proof}
That these are closed sets follows from Lemma~\ref{lem:gaps_are_continuous}.
For large $\vec{\lambda}$, we know
\begin{equation*}
\left\Vert L_{\vec{\lambda}}(\mathbb{A})-L_{\vec{\lambda}}(0,\dots,0)\right\Vert  =
\left\Vert L_{\vec{0}}(\mathbb{A})\right\Vert
\end{equation*}
and so
\begin{equation*}
\mu \left( L_{\vec{\lambda}}(\mathbb{A})\right) \geq  |\vec{\lambda}| -\left\Vert L_{\vec{0}}(\mathbb{A})\right\Vert .
\end{equation*}
This proves that each $\Lambda_{\epsilon}^{\mathrm{C}}(\mathbb{A})$
is also a bounded set.

For the quadratic case, we need to use 
\begin{equation*}
M_{\vec{\lambda}}(\mathbb{A})=\left[\begin{array}{c}
A_{1}-\lambda_{1}I\\
\vdots\\
A_{d}-\lambda_{d}I
\end{array}\right].
\end{equation*}
Since $M_{\vec{\lambda}}(\mathbb{A})^{*}M_{\vec{\lambda}}(\mathbb{A})=Q_{\vec{\lambda}}(\mathbb{A})$
we know $\mu_{\vec{\lambda}}^{Q}(\mathbb{A})=s_{\min}\left(M_{\vec{\lambda}}(\mathbb{A})\right)$
where now we use the theory of $s_{\min}$ as in \cite{cerjan2025multivariable_PS}. Now we find  
\begin{equation*}
\left\Vert M_{\vec{\lambda}}(\mathbb{A})-M_{\vec{\lambda}}(0,\dots,0)\right\Vert \leq \left\Vert M_{\vec{0}}(\mathbb{A})\right\Vert.
\end{equation*}
These sets are bounded also in the quadratic case, using now \cite[Lemma 4.3] {cerjan2025multivariable_PS}.
\end{proof}

The next two results show that the quadratic pseudospectrum is very closely connected to the problem of finding states that make the sum of variance small.

\begin{thm}
\label{thm:quad_PS_as_lower_bound}
Let $\mathcal{A}$ be a unital
$C^{*}$-algebra, let $\mathbb{A}:=(A_{1},\dots,A_{d})$ be a $d$-tuple
of elements from $\mathcal{A}$, and let $\psi\in\mathcal{S}(\mathcal{A})$
be a state. Then 
\begin{equation*}
\sum_{j=1}^{d}\Delta_{\psi}^{2}(A_{j})\geq\left(\mu_{\psi(\mathbb{A})}^{\mathrm{Q}}(\mathbb{A})\right)^{2}.
\end{equation*}
\end{thm}

\begin{proof}
Fix $\psi\in\mathcal{S}(\mathcal{A})$ and set $Q_{\psi}(\mathbb{A}):=Q_{\psi(\mathbb{A})}(\mathbb{A})$ and $\mu_\psi(\mathbb{A}):=\mu_{\psi(\mathbb{A})}^{\mathrm{Q}}(\mathbb{A}).$ Then
\begin{equation}
\psi(Q_{\psi}(\mathbb{A}))
=
\sum\limits _{j=1}^{d}\psi\left((A_{j}-\psi(A_{j}))^{*}(A_{j}-\psi(A_{j}))\right)
=
\sum\limits _{j=1}^{d}\Delta_{\psi}^{2}(A_{j}).\label{eq:psi_on_Q_is_sum_uncertainty}
\end{equation}
By Definition~\ref{def:quadratic_gap} and Lemma~\ref{lem:bounds_on_state_values}, we have $(\mu_{\psi}^{\mathrm{Q}}(\mathbb{A}))^{2}
 \leq
 \psi(Q_{\psi}(\mathbb{A}))$. 
\end{proof}

\begin{thm}
\label{thm:min_of QPS_is_min_uncertaintly}Let $\mathcal{A}$ be a unital
$C^{*}$-algebra and let $\mathbb{A}:=(A_{1},\dots,A_{d})$ be a $d$-tuple
of elements from $\mathcal{A}$. 
\begin{equation*}
\min\left\{ \left.\mu_{\vec{\lambda}}^{\mathrm{Q}}(\mathbb{A})\right|\vec{\lambda}\in\mathbb{C}^{d}\right\} 
=
\min
\Bigl\{ 
\Bigl.
\Bigl(
\sum_{j=1}^{d}\Delta_{\psi}^{2}(A_{j})
\Bigr)^{1/2} 
\Bigr| 
\psi\in\mathcal{S}(\mathcal{A})
\Bigr\}.
\end{equation*}
\end{thm}

\begin{proof}
Let $\vec{\nu}$ satisfy
\begin{equation*}
\mu_{\vec{\nu}}^{\mathrm{Q}}(\mathbb{A})
=
\min\left\{ \left.\mu_{\vec{\lambda}}^{\mathrm{Q}}(\mathbb{A})\right|\vec{\lambda}\in\mathbb{C}^{d}\right\} .
\end{equation*}
By Lemma~\ref{lem:bounds_on_state_values}, there exists $\psi\in\mathcal{S}\left(\mathcal{A}\right)$ such that
\begin{equation*}
\psi(Q_{\vec{\nu}}(\mathbb{A}))=\left(\mu_{\vec{\nu}}^{\mathrm{Q}}(\mathbb{A})\right)^{2}.
\end{equation*}
Then
\begin{align*}
\psi(Q_{\vec{\nu}}(\mathbb{A})) 
& =
\sum\limits_{j=1}^d\psi(A_{j}^{*}A_{j})-\nu_{j}\overline{\psi(A_j)}-\overline{\nu_{j}}\psi(A_j)+\overline{\nu_{j}}\nu_{j}I\\
 & =\sum\limits_{j=1}^d\Delta_{\psi}^{2}(A_{j})+\left|\,\left|\psi(A_{j})\right|-\nu_{j}I\right|^{2}.
\end{align*}
Hence,
\begin{equation*}
\psi(Q_{\vec{\nu}}(\mathbb{A}))
\geq
\sum\limits_{j=1}^d\Delta_{\psi}^{2}(A_{j}),
\end{equation*}
and so 
\begin{equation*}
\min\left\{ \left.\mu_{\vec{\lambda}}^{\mathrm{Q}}(\mathbb{A})\right|\vec{\lambda}\in\mathbb{C}^{d}\right\} 
=
\Bigl(\psi(Q_{\vec{\nu}}(\mathbb{A}))\Bigr)^{\frac{1}{2}}
\geq
\min\Bigl\{ \Bigl(\sum_{j=1}^{d}\Delta_{\varphi}^{2}(A_{j})\Bigr)^{\frac{1}{2}}\Bigr|\varphi\in\mathcal{S}(\mathcal{A})\Bigr\} .
\end{equation*}

Now let's work on the other needed inequality. Let $\psi$ be a state on $\mathcal{A}$ such that 
\begin{equation*}
\sum_{j=1}^{d}\Delta_{\psi}^{2}(A_{j})
=
\min\Bigl\{ \sum_{j=1}^{d}\Delta_{\varphi}^{2}(A_{j})\Bigr|\varphi\in\mathcal{S}(\mathcal{A})\Bigr\} .
\end{equation*}
We find 
\begin{equation*}
\left(\mu_{\psi}^{\mathrm{Q}}(\mathbb{A})\right)^{2}
\leq
\psi(Q_\psi(\mathbb{A}))
=
\sum_{j=1}^{d}\Delta_{\psi}^{2}(A_{j})
=
\min\left\{ \sum_{j=1}^{d}\Delta_{\varphi}^{2}(A_{j})\Bigr|\varphi\in\mathcal{S}(\mathcal{A})\right\} ,
\end{equation*}
which means 
\begin{equation*}
\min\left\{ \left.\left(\mu_{\vec{\lambda}}^{\mathrm{Q}}(\mathbb{A})\right)^{2}\right|\vec{\lambda}\in\mathbb{C}^{d}\right\} 
\leq
\min\left\{ \sum_{j=1}^{d}\Delta_{\varphi}^{2}(A_{j})\Bigr|\varphi\in\mathcal{S}\left(\mathcal{A}\right)\right\} .
\end{equation*}
\end{proof}

A special case of Theorem~\ref{thm:min_of QPS_is_min_uncertaintly} was found previously in \cite{schwonnek2017UC_in_noisy_systems}. Computing the minimum of the quadratic local gap can be an effective way to find a state-independent lower bound on uncertainty.
Looking back at the proof, we can see that every state which minimizes the sum of variances must also give minimal expectation  for the quadratic composite operator at a location that minimizes the quadratic gap. The converse also follows from re-examining the proof. Thus, we have the following. 

\begin{thm}
\label{thm:Finding_max_localized_states}The set of states $\psi$
on $\mathcal{A}$ that minimize $\sum\limits_{j=1}^d\Delta_{\psi}^{2}(A_{j})$
equals 
\begin{equation*}
\left\{ \psi\in\mathcal{S}(\mathcal{A})\left|\psi\left(\mu_{\vec{\lambda}}^{\mathrm{Q}}(\mathbb{A})\right)=\mu\left(Q_{\vec{\lambda}}(\mathbb{A})\right)\text{ and }\vec{\lambda}\in\Omega(\mathbb{A})\right.\right\},
\end{equation*}
where 
\begin{equation*}
\Omega(\mathbb{A})=\left\{ \vec{\lambda}\in\mathbb{C}^{d}\left|\mu_{\vec{\lambda}}^{\mathrm{Q}}(\mathbb{A})\text{ is a global minimum}\right.\right\} .
\end{equation*}
\end{thm}

\begin{rem}
In many cases, Theorem~\ref{thm:Finding_max_localized_states} is
sufficient for using numerical methods to identify all the states
that are maximally localized.  The issue is that
when $A_{j}$ are very large matrices, computing $\Omega(\mathbb{A})$
will be difficult. Part of the difficulty is that we don't know \emph{a
priori} the global minimum of the quadratic gap. Since we find, in
practice, that the global minimum of the Clifford gap is always zero,
it can be faster to first compute the compact set where the Clifford
gap is small, and then compute the quadratic gaps only on that compact
set.  
\end{rem}

\section{The main theorem} \label{sec:results_on_uncertainty}

\begin{lem}
\label{lem:L-Q-R}
If $\mathbb{A}:=(A_{1},\dots,A_{d})$ is a $d$-tuple of elements
from $\mathcal{A}$, a unital $C^{*}$-algebra then
\begin{equation*}
L(\mathbb{A})^{*}L(\mathbb{A})=Q(\mathbb{A})\otimes I_{N}+\sum_{j<k}\left[A_{j},A_{k}\right]_{{\textstyle *}}\otimes\Gamma_{j}\Gamma_{k}\label{eq:L_Q_R}
\end{equation*}
 where $N$ is as in Definition~\ref{def:spectral_localizer}.
\end{lem}

\begin{proof}
The product 
\begin{equation*}
L(\mathbb{A})^{*}L(\mathbb{A})
=
\left(\sum_{j}A_{j}^{*}\otimes\Gamma_{j}\right)\left(\sum_{k}A_{k}\otimes\Gamma_{k}\right)
\end{equation*}
is the sum of 
\begin{equation*}
\sum_{j}A_{j}^{*}A_{j}\otimes\Gamma_{j}^{2}
\end{equation*}
and
\begin{equation*}
\sum_{j<k}\left(A_{j}^{*}A_{k}\otimes\Gamma_{j}\Gamma_{k}+A_{k}^{*}A_{j}\otimes\Gamma_{k}\Gamma_{j}\right).
\end{equation*}
Since $\Gamma_{j}^{2}=I_{N}$ for all $1\leq j\leq d$, the first summand equals
$Q(\mathbb{A})\otimes I_{N}$.  Since
 $\Gamma_{k}\Gamma_{j}=-\Gamma_{j}\Gamma_{k}$ for all $1\leq i,j\leq d$, the second summand simplifies to
\begin{equation*}
\sum_{j<k}\left(A_{j}^{*}A_{k}\otimes\Gamma_{j}\Gamma_{k}+A_{k}^{*}A_{j}\otimes\Gamma_{k}\Gamma_{j}\right)
=
\sum_{j<k}\left(A_{j}^{*}A_{k}-A_{k}^{*}A_{j}\right)\otimes\Gamma_{j}\Gamma_{k}.
\end{equation*}
\end{proof}

Now, given a product state $\psi\otimes\varphi$, with arbitrary $\psi\in\mathcal{S}(\mathcal{A})$ and $\varphi\in\mathcal{S}(\mathbf{M}_{n})$, let's take the expected value of each of the terms in (\ref{eq:L_Q_R}).
We find that
\begin{equation}
\left(\psi\otimes\varphi\right)\left(Q(\mathbb{A})\otimes I_{N}\right)
=
\psi\left(Q(\mathbb{A})\right)\label{eq:product_state_on_Q},
\end{equation}
which we know  relates to the sum of variances
with respect to $\psi$, by Theorem ~\ref{thm:quad_PS_as_lower_bound}. Also, 
\begin{equation*}
\left(\psi\otimes\varphi\right)\left(\left(L(\mathbb{A})\right)^{*}\left(L(\mathbb{A})\right)\right)
\geq
\mu_{\vec{0}}^{\mathrm{C}}(\mathbb{A})\label{eq:product_state_on_L},
\end{equation*}
so that is promising since the Clifford pseudospectrum is now understood in many settings in physics. We now examine the third term in (\ref{eq:L_Q_R}) and discover
\begin{equation} \label{eq:product_state_on_R}
\begin{split}
\left(\psi\otimes\varphi\right)\left(\sum_{j<k}\left[A_{j},A_{k}\right]_{{\textstyle *}}\otimes\Gamma_{j}\Gamma_{k}\right) & =\sum_{j<k}\psi\left(i\left[A_{j},A_{k}\right]_{{\textstyle *}}\right)\varphi\left(-i\Gamma_{j}\Gamma_{k}\right)
\\
 & =\varphi\left(\sum_{j<k}\psi\left(i\left[A_{j},A_{k}\right]_{{\textstyle *}}\right)\left(-i\Gamma_{j}\Gamma_{k}\right)\right).
\end{split}
\end{equation}

\begin{defn}
\label{def:R_psi-matrix} If $\mathbb{A}:=(A_{1},\dots,A_{d})$ is
a $d$-tuple of elements from $\mathcal{A}$, a unital $C^{*}$-algebra,
and $\psi\in\mathcal{S}(\mathcal{A})$ then we define $R_{\psi}(\mathbb{A})$
as the matrix
\begin{equation*}
R_{\psi}(\mathbb{A})
:=
\sum_{j<k}\psi\left(i\left[A_{j},A_{k}\right]_{{\textstyle *}}\right)\left(-i\Gamma_{j}\Gamma_{k}\right).
\end{equation*}
\end{defn}

This will always be Hermitian, regardless of hermiticity of the $A_j$. In the appendix we show that the spectrum $R_\psi(\mathbb{A})$ is always symmetric about the origin. For $d=3$, the spectrum contains just the two points
\begin{equation*}
\pm\sqrt{\alpha_{12}^{2}+\alpha_{13}^{2}+\alpha_{23}^{2}},
\end{equation*}
where $\alpha_{jk}=\psi\left(i\left[A_{j},A_{k}\right]_{{\textstyle *}}\right)$.
For $d=4$, the inelegant spectrum is now the four points 
\begin{equation*}
\pm\sqrt{\alpha_{12}^{2}+\alpha_{13}^{2}+\alpha_{14}^{2}+\alpha_{23}^{2}+\alpha_{24}^{2}+\alpha_{34}^{2}+2\alpha_{12}\alpha_{34}-2\alpha_{13}\alpha_{24}+2\alpha_{14}\alpha_{23}}
\end{equation*}
and
\begin{equation*}
\pm\sqrt{\alpha_{12}^{2}+\alpha_{13}^{2}+\alpha_{14}^{2}+\alpha_{23}^{2}+\alpha_{24}^{2}+\alpha_{34}^{2}-2\alpha_{12}\alpha_{34}+2\alpha_{13}\alpha_{24}-2\alpha_{14}\alpha_{23}}.
\end{equation*}
The formulas here can easily be derived with a computer algebra system, but the mess compounds as $d$ increases (see the supplemental file \texttt{R\_trace.nb}) \cite{Loring_Uncertainty_code_2026}.
This cleans up in some special cases, but in general we just use the following lemmas whose proofs appear in the appendix. Both lemmas are in regard to a general matrix of the form
\begin{equation}
R=\sum_{j<k}\alpha_{jk}\left(-i\Gamma_{j}\Gamma_{k}\right)\label{eq:Def_of_R}
\end{equation}
for arbitrary real scalars $\{\alpha_{jk}\}_{j<k}$.

\begin{lem}
\label{lem:R-matrix_when_zero} Let $R$ be as in (\ref{eq:Def_of_R}).
Then $R=0$ if, and only if, all the scalars $a_{jk}$ are zero.
\end{lem}

\begin{lem}
\label{lem:R_matrix_symmetric_spectrum} If $R$ is as in (\ref{eq:Def_of_R}), then the spectrum of $R$ is symmetric about $0$. In particular, $-\|R\|$ is always an eigenvalue of $R$.
\end{lem}

The main theorem will follow from the following lemma.

\begin{lem}
\label{lem:Techincal_Inequality} 
Let $\mathcal{A}$ be a unital $C^{*}$-algebra, let $\mathbb{A}:=(A_{1},\dots,A_{d})$ be a $d$-tuple of elements
from $\mathcal{A}$, and let $\vec{\lambda}\in\mathbb{C}^{d}$, and
let $\psi\in\mathcal{S}(\mathcal{A})$ be a state. Let
\begin{equation*}
\mathbb{A}-\vec{\lambda}I=(A_{1}-\lambda_{1}I,\dots,A_{d}-\lambda_{d}I).
\end{equation*}
If $R_{\psi}(\mathbb{A}-\vec{\lambda}I)$ is as in Definition~\ref{def:R_psi-matrix} then
\begin{equation}
\psi\left(Q_{\vec{\lambda}}\left(\mathbb{A}\right)\right)\geq\left(\mu_{\vec{\lambda}}^{\mathrm{C}}(\mathbb{A})\right)^{2}+\bigl\Vert R_{\psi}(\mathbb{A}-\vec{\lambda}I)\bigr\Vert.
\label{eq:general_equation}
\end{equation}
\end{lem}

\begin{proof}
For any state state $\varphi$ on $\mathbf{M}_{N}$ we have
\begin{equation*}
\psi\left(Q_{\vec{\lambda}}\left(\mathbb{A}\right)\right)=\left(\psi\otimes\varphi\right)\left(Q_{\vec{\lambda}}\left(\mathbb{A}\right)\otimes I\right).
\end{equation*}
By Lemma~\ref{lem:L-Q-R} we have
\begin{equation*}
\psi\left(Q_{\vec{\lambda}}\left(\mathbb{A}\right)\right)=\left(\psi\otimes\varphi\right)\left(L_{\vec{\lambda}}\left(\mathbb{A}\right)^{*}L_{\vec{\lambda}}\left(\mathbb{A}\right)^{*}\right)-\left(\psi\otimes\varphi\right)\left(\sum_{j<k}\bigl[A_{j}-\lambda_{j}I,A_{k}-\lambda_{k}I\bigr]_{{\textstyle *}}\otimes\Gamma_{j}\Gamma_{k}\right).
\end{equation*}
For any $\varphi$ we have
\begin{equation*}
\left(\psi\otimes\varphi\right)\left(L_{\vec{\lambda}}\left(\mathbb{A}\right)^{*}L_{\vec{\lambda}}\left(\mathbb{A}\right)^{*}\right)\geq\left(\mu_{\vec{\mathbf{\lambda}}}^{\mathrm{C}}(\mathbb{A})\right)^{2}.
\end{equation*}
Also, by (\ref{eq:product_state_on_R}) we have
\begin{equation*}
\left(\psi\otimes\varphi\right)\left(\sum_{j<k}\bigl[A_{j}-\lambda_{j}I,A_{k}-\lambda_{k}I\bigr]_{{\textstyle *}}\otimes\Gamma_{j}\Gamma_{k}\right)=\varphi\left(R_{\psi}(\mathbb{A}-\vec{\lambda}I)\right).
\end{equation*}
There is a choice of $\varphi$ such that 
\begin{equation*}
\left(\psi\otimes\varphi\right)\left(\sum_{j<k}\bigl[A_{j}-\lambda_{j}I,A_{k}-\lambda_{k}I\bigr]_{{\textstyle *}}\otimes\Gamma_{j}\Gamma_{k}\right)=-\left\Vert R_{\psi}(\mathbb{A}-\vec{\lambda}I)\right\Vert 
\end{equation*}
where 
\begin{equation*}
\mathbb{A}-\vec{\lambda}I=(A_{1}-\lambda_{1}I,\dots,A_{d}-\lambda_{d}I).
\end{equation*}
This proves (\ref{eq:general_equation}).
\end{proof}

\begin{thm}
\label{thm:Main_result} 
Let $\mathcal{A}$ be a unital $C^{*}$-algebra,
let $\mathbb{A}:=(A_{1},\dots,A_{d})$ be a $d$-tuple of elements
from $\mathcal{A}$, and let $\vec{\lambda}\in\mathbb{C}^{d}$, and
let $\psi\in\mathcal{S}(\mathcal{A})$ be a state. Let
\begin{equation*}
\mathbb{A}-\vec{\lambda}I=(A_{1}-\lambda_{1}I,\dots,A_{d}-\lambda_{d}I).
\end{equation*}
With these assumptions, 
\begin{equation*}
\sum_{j=1}^{d}\Delta_{\psi}^{2}(A_{j})\geq\left(\mu_{\psi(\mathbb{A})}^{\mathrm{C}}(\mathbb{A})\right)^{2}+\bigl\Vert R_{\psi}(\hat{\mathbb{A}})\bigr\Vert.
\end{equation*}
and
\begin{equation*}
\left(\mu_{\vec{\lambda}}^{\mathrm{Q}}(\mathbb{A})\right)^{2}\geq\left(\mu_{\vec{\lambda}}^{\mathrm{C}}(\mathbb{A})\right)^{2}+\bigl\Vert R_{\psi}(\mathbb{A}-\vec{\lambda}I)\bigr\Vert .
\end{equation*}
\end{thm}

\begin{proof}
These both follow from Lemma~\ref{lem:Techincal_Inequality}. To
prove the first inequality, we let $\vec{\lambda}=\psi(\mathbb{A})$
for the given state $\psi$. Then by Lemma~\ref{thm:quad_PS_as_lower_bound} we have
\begin{equation*}
\psi\left(Q_{\psi(\mathbb{A})}\left(\mathbb{A}\right)\right)=\sum_{j=1}^{d}\Delta_{\psi}^{2}(A_{j})
\end{equation*}
 so the result follows. For the second, we are given $\vec{\lambda}$
and can select $\psi$ so that 
\begin{equation*}
\psi\left(Q_{\vec{\lambda}}\left(\mathbb{A}\right)\right)=\mu\left(Q_{\vec{\lambda}}\left(\mathbb{A}\right)\right)=\left(\mu_{\vec{\lambda}}^{\mathrm{Q}}(\mathbb{A})\right)^{2}.
\end{equation*}
\end{proof}

\section{Special cases}

Both summands in the lower-bounds in our main theorem simplify when
restricted to the Hermitian case. For the summand involving the Clifford pseudospectrum, we don't yet know how hard it will be to calculate this in practice, except when computing numerically. The Clifford gap is not inherently easier or harder to compute than the quadratic gap for a fixed $\vec{\lambda}$.
However, it is easier to compute $\Lambda_{\epsilon}^{\mathrm{C}}\left(\mathbb{A}\right)$
for small $\epsilon$ than to compute $\Lambda_{\delta+\epsilon}^{\mathrm{Q}}\left(\mathbb{A}\right)$
where $\delta^2$ is the global minimum of the quadratic spectrum, 
\begin{equation*}
\delta^{2}=\min_{\|\mathbf{v}\|=1}\sum_{j}\Delta_{\mathbf{v}}^{2}(A_{j}).
\end{equation*}
Computing a crude approximation of the low-lying Clifford gap first
can potentially speed up computing a good version of the low-lying quadratic gap. 

The Hermitian case will be our focus, for the most part. There is
growing literature  \cite{Cerjan2024classifying,saito2025gapless_quasicrystals,steinacker2021QuantumGeometry}
where the Clifford spectrum or pseudospectrum is calculated, mainly in the Hermitian case. Here also we can visualize the various pseudospectral sets involving three observables, as the relevant domain of the Clifford, or quadratic, gap is now $\mathbb{R}^{3}$. 

The second summand also simplifies in the Hermitian case. With real
scalars and Hermitian observables, we have
\begin{equation*}
\left[A_{j}-\lambda_{j}I,A_{k}-\lambda_{k}I\right]_{{\textstyle *}}=\left[A_{j},A_{k}\right].
\end{equation*}
The second summand  here the constant $\left\Vert R_{\psi}\left(\mathbb{A}\right)\right\Vert $.
In numerical settings, this is a fast calculation, but in theoretical
settings there is just no reasonable formulas beyond the three-observable
setting. We can do better in many special cases.  The proof of the following is in the appendix.

\begin{lem} \label{lem:R_in_special_cases}
In the Hermitian setting, the following hold: 
\begin{itemize}
    \item For $d=2$, 
\begin{equation*}
\norm{R_{\psi}(\mathbb{A})} =
|\psi(i[A_{1},A_{2}])|.
\end{equation*}
\item For $d=3$, 
\begin{equation*}
\norm{R_{\psi}(\mathbb{A})} =\sqrt{\left|\psi\left(i\left[A_{1},A_{2}\right]\right)\right|^{2}+\left|\psi\left(i\left[A_{1},A_{3}\right]\right)\right|^{2}+\left|\psi\left(i\left[A_{2},A_{3}\right]\right)\right|^{2}}.
\end{equation*}
\item For $d\geq 4$ and $A_{j}A_{d}=A_{d}A_{j}$ for all $j<d$,
\begin{equation*}
\norm{R_{\psi}(\mathbb{A})}
=
\sqrt{\sum_{j=1}^{d-1}\left|\psi\left(i\left[A_{j},A_{d}\right]\right)\right|^{2}}.
\end{equation*}
\end{itemize}
\end{lem}

\section{Examples}

\begin{example}\label{ex:angular_momenta}
Angular momentum provides a fundamental example of incompatible quantum observables, with applications ranging from atomic spectroscopy to quantum spin systems. Given $(J_x, J_y, J_z)$ that satisfy $[J_i,J_j] = i \hbar \varepsilon_{ijk} J_k$. Note 
\begin{equation*}
(\Delta J_i)^2 \equiv \textrm{Var}(J_i) 
= 
\langle J_i^2\rangle - \langle J_i \rangle^2.
\end{equation*}

We need to be able to show $(\Delta J_x)^2 + (\Delta J_y)^2 + (\Delta J_z)^2 \ge j\hbar^2$. Ordinarily, one would show this inequality using
\begin{align}
    (\Delta J_x)^2 + (\Delta J_y)^2 + (\Delta J_z)^2 & = \langle J_x^2 + J_y^2 + J_z^2 \rangle - \left( \langle J_x \rangle^2 + \langle J_y \rangle^2 + \langle J_z \rangle^2 \right) \notag \\
    &= \langle \mathbf{J}^2 \rangle - |\langle \mathbf{J} \rangle |^2
\end{align}
and then noting
\begin{equation}
    \langle \mathbf{J}^2 \rangle = j(j+1)\hbar^2
\end{equation}
where a state in the total angular momentum basis has $|j,m\rangle$, where $j \in \mathbb{N}$ is the total angular momentum quantum number and $m \in \{-j,-j+1,\dots,j-1,j\}$, $J_z |j,m\rangle = m\hbar $. The lower bound for the sum of the variances is found using the maximum possible value of $|\langle \mathbf{J} \rangle |^2$, denote as $|\langle \mathbf{J} \rangle |^2_{\textrm{max}}$, which occurs when $m = \pm j$. Then,
\begin{equation}
    |\langle \mathbf{J} \rangle |^2_{\textrm{max}} = j^2\hbar^2.
\end{equation}
Thus,
\begin{equation}
    (\Delta J_x)^2 + (\Delta J_y)^2 + (\Delta J_z)^2 \ge j \hbar^2.
\end{equation}

Using Lemma~\ref{lem:R_in_special_cases}, we obtain
\begin{align}
(\Delta J_x)^2 + (\Delta J_y)^2 + (\Delta J_z)^2
& \ge \sqrt{\left\langle i\left[J_x,J_y\right]\right\rangle^{2} + \left\langle i\left[J_x,J_z\right]\right\rangle^{2} + \left\langle i\left[J_y,J_z\right] \right\rangle ^{2}} \notag \\
& = \sqrt{\left\langle i (i \hbar J_z)\right\rangle^{2} + \left\langle i(-i\hbar J_y)\right\rangle^{2} + \left\langle i(i \hbar J_x) \right\rangle ^{2}} \notag \\
& = \hbar \sqrt{\left\langle J_z\right\rangle^{2} + \left\langle J_y\right\rangle^{2} + \left\langle J_x \right\rangle ^{2}} \notag \\
& = \hbar |\langle \mathbf{J} \rangle |.
\end{align}
Thus, the bound this approach reproduces is sharp for spin-coherent states, for which $|\langle \mathbf{J} \rangle| = j\hbar$. 

Alternatively, the quadratic pseudospectral gap provides a stronger bound. Setting $\vec{\lambda} = \langle \mathbf{J} \rangle$, the quadratic composite operator becomes
\[
Q_{\langle \mathbf{J}\rangle}(\mathbf{J})
=j(j+1)\hbar^2I
-2 \langle \mathbf{J}\rangle \cdot\mathbf J
+|\langle \mathbf{J}\rangle|^2I.
\]
Since the largest eigenvalue of \(\langle \mathbf{J}\rangle \cdot\mathbf J\) is \(j\hbar|\langle \mathbf{J}\rangle|\), we obtain
\[
\bigl(\mu^Q_{\langle \mathbf{J} \rangle}(\mathbf J)\bigr)^2
=j\hbar^2+(j\hbar-|\langle \mathbf{J} \rangle|)^2.
\]
Thus, Theorem~\ref{thm:quad_PS_as_lower_bound} yields
\[
\sum_i(\Delta J_i)^2
\geq j\hbar^2+
(j\hbar-|\langle\mathbf J\rangle|)^2
\geq j\hbar^2,
\]
recovering the conventional angular-momentum uncertainty relation, with equality for spin-coherent states.
\end{example}


\begin{example}
This example goes into a subcase of Example~\autoref{ex:angular_momenta} in order to compare our quadratic gap lower bound on uncertainty, rather than the lower bound provided by the $R$-matrix. Let $\sigma_x,\sigma_y,$ and $\sigma_z$ be the Pauli matrices in $\mathbf{M}_2$, and let $\vec{\sigma}:=(\sigma_x,\sigma_y,\sigma_z)$. For a given state $\psi\in \mathcal{S}(\mathbf{M}_2),$  define the vector of expectation values \(\vec{\lambda} = (\psi(\sigma_x), \psi(\sigma_y), \psi(\sigma_z))\), which map to a point inside or on the unit Bloch sphere. As the Pauli matrices describe the three spin components of a spin-$\tfrac{1}{2}$ particle (up to a factor of $\tfrac{\hbar}{2}$), bounds on their combined variance quantify the unavoidable uncertainty in simultaneously characterizing all three components of a two-level quantum state, including mixed states. Then the quadratic composite operator is given by
 \begin{equation*}
 Q_{\vec{\lambda}}(\vec{\sigma})
 =
 Q_{\psi }(\vec{\sigma})
 =
 \sum _{j\in \{x,y,z\}}(\sigma _{j}-\psi (\sigma _{j})I_2)^{2}
 =(3+\|\vec{\lambda }\|{}^{2})I_2-2(\vec{\lambda }\cdot \vec{\sigma }),
 \end{equation*}
 and the eigenvalues of $Q_{\vec{\lambda}}(\vec{\sigma})$ are readily found to be 
 \begin{equation*}
 \lambda_{\pm} = 3 + \|\vec{\lambda}\|^2 \pm 2\|\vec{\lambda}\|.
 \end{equation*} 
 Factoring the lower eigenvalue $\lambda_{-}$ yields the quadratic gap:
 \begin{equation*}
\left( \mu^{\mathrm{Q}}_{\psi}(\mathbb{A})\right)^2
 =
 (\|\vec{\lambda }\|-1)^{2}+2.
 \end{equation*}
 Therefore, for any state $\psi\in \mathcal{S}(\mathbf{M}_2)$, Theorem~\ref{thm:Main_result} gives a state-dependent lower bound on the sum of variances:
 \begin{equation}
     \displaystyle \sum\limits_{j\in \{x,y,z\}} \Delta^2_\psi(\sigma_j) 
     \geq
     \left(\mu^{\mathrm{Q}}_\psi(\mathbb{A})\right)^2 \geq 2.
 \end{equation}
 In particular, when $\psi$ is pure $\|\vec{\lambda}\| = 1$, our lower bound corresponds to the known uncertainty of $\sqrt{2}$, and when $\psi$ is maximally mixed ($\|\vec{\lambda}\|= 0$), the quadratic gap $\left(\mu_{\psi}^{\mathrm{Q}}(\vec{\sigma})\right)^2$ is $3$, which saturates the known uncertainty of $\sqrt{3}.$
\end{example}

\begin{example}
 Let $ \ell^2(\mathbb{N})$ denote the collection of square-summable functions $f:\mathbb{N}\to \mathbb{C}$ where $\mathbb{N}$ denotes the set of nonnegative integers $\{0,1,2,\dots\}$, and let $\{\delta_n:n\in \mathbb{N}\}$ denote the orthonormal basis for $\ell^2(\mathbb{N})$ defined for each $n\in \mathbb{N}$ by 
    \begin{equation*}
    \delta_n(m)
    =
    \begin{cases}
        1 &\mbox n=m \\
        0 &\mbox {otherwise}
    \end{cases}.
    \end{equation*}
    Define the unilateral shift $S:\ell^2(\mathbb{N})\to \ell^2(\mathbb{N})$ by $S(\delta_n) = \delta_{n+1}$ for all $n\in \mathbb{N}$. The unilateral shift also has a natural interpretation as a unidirectional hopping operator on a semi-infinite one-dimensional lattice, where the basis states $\delta_n$ represent localized lattice sites. Its adjoint describes hopping in the opposite direction, and the special role of the boundary at $n=0$ makes this a simple setting for examining how spatial boundaries affect operator uncertainty. For $n\geq 1$, one can compute that the adjoint $S^*$ of $S$ satisfies
    \begin{equation*}
    S^*(\delta_n )
    =
    \begin{cases}
    \delta_{n-1} &\mbox n\geq 1\\ 
    0 &\mbox n=0
    \end{cases}.
    \end{equation*}
    Then $C^*(S)$ is a representation of the Toeplitz algebra $\mathcal{T}_z$, the universal $C^*$-algebra generated by a non-unitary isometry, and Coburn's Theorem implies that $C^*(S)\cong \mathcal{T}_z$.


    We show here how well Theorem~\autoref{thm:Main_result} bounds uncertainty for the non-normal elements $S$ and $S^*$ in $C^*(S)$ at a pure state given by a vector $\xi=\sum_{n=0}^\infty \lambda_n\delta_n\in \ell^2(\mathbb{N})$. First, we calculate a general formula for the sum of variances of $S$ and $S^*$ with respect to $\xi$:
    \begin{align*}
    \Delta^2_{\xi}(S) +\Delta^2_\xi(S^*)
    &=
    \|S\xi\|^2 - |\langle S\xi,\xi\rangle|^2 
    + \|S^*\xi\|^2 - |\langle S^*\xi,\xi\rangle|^2\\
    &=
    1 - \left|\sum_{m=0}^\infty \lambda_{m+1}\overline{\lambda_m}\right|^2
    + (1-|\lambda_0|^2) - \left|\sum_{m=1}^\infty \lambda_{m-1}\overline{\lambda_m} \right|^2\\
     &=
    2 - |\lambda_0|^2 - 2\left|\sum_{m=0}^\infty \lambda_{m+1}\overline{\lambda_m}\right|^2.
    \end{align*}

\begin{enumerate}
    \item For $\xi=\delta_0$, we get precisely 
\begin{equation*}
\Delta^2_{\delta_0}(S) + \Delta^2_{\delta_0}(S^*)
=
2-|\lambda_0|^2
=
1.
\end{equation*}
Observe that the pure state arising from $\delta_0$ applied to the quadratic composite operator $Q_{\delta_0}(S,S^*) = S^*S + SS^* = I+SS^*$ yields
\begin{equation*}
\langle Q_{\delta_0}(S,S^*)\delta_0,\delta_0\rangle
=
\langle \delta_0,\delta_0\rangle
+
\langle SS^*\delta_0,\delta_0\rangle
=
\norm{\delta_0}^2
=
1,
\end{equation*}
which says that the quadratic gap is a sharp lower bound sharp for $S$ and $S^*$. On the other hand, because $\langle S\delta_0,\delta_0\rangle = \langle S^*\delta_0,\delta_0\rangle = 0$, the lower bound in Theorem~\ref{thm:Main_result} simplifies to 
\begin{equation*}
\Delta^2_{\delta_0}(S)+\Delta^2_{\delta_0}(S^*)\geq |\langle S^*S^*\delta_0,\delta_0\rangle - \langle SS\delta_0,\delta_0\rangle|.
\end{equation*}
In this case, the lower bound arising from the $R_{\delta_0}$-matrix is identically 0.
\item Now consider the pure state arising from the unit vector $\psi_{0,2}:=\frac{1}{\sqrt{2}}(\delta_0+i\delta_2)$. We again get that the sum of variances is
\begin{equation*}
\Delta^2_{\psi_{0,2}}(S) + \Delta^2_{\psi_{0,2}}(S^*)
=
2-|\lambda_0|^2
=
2-\left|\frac{1}{\sqrt{2}}\right|^2
=
\frac{3}{2}.
\end{equation*}
Note $\langle SS^*\delta_0,\delta_2\rangle = \langle SS^*\delta_2,\delta_0\rangle = 0,$ so
\begin{equation*}
\langle SS^*\psi_{0,2},\psi_{0,2}\rangle
=
\frac{1}{2}
\left(
\langle S^*\delta_0,S^*\delta_0\rangle 
+ \langle S^*\delta_2,S^*\delta_2 \rangle
\right)
=
\frac{1}{2}(0+\norm{\delta_1}^2) =\frac{1}{2},
\end{equation*}
while 
\begin{equation*}
\langle S^*S \psi_{0,2},\psi_{0,2}\rangle 
=
\langle \psi_{0,2},\psi_{0,2} \rangle
=
\norm{\psi_{0,2}}^2
=1.
\end{equation*}
Hence, the lower bound on $\Delta^2_{\psi_{0,2}}(S)+\Delta^2_{\psi_{0,2}}(S^*)$ obtained from the quadratic composite operator in Theorem~\ref{thm:Main_result} is
\begin{align*}
\Delta^2_{\psi_{0,2}}(S)+\Delta^2_{\psi_{0,2}}(S^*)
&\geq 
\langle Q_{\psi_{0,2}}(S,S^*)\psi_{0,2},\psi_{0,2}\rangle\\
&=
\langle S^*S \psi_{0,2},\psi_{0,2}\rangle + \langle SS^* \psi_{0,2},\psi_{0,2}\rangle\\
&=
\frac{3}{2}.
\end{align*}
Now we check to see what lower bound is provided by the $R$-matrix in Lemma~\ref{lem:R_in_special_cases}. Since $S^*S^*(\psi_{0,2}) = \frac{i}{\sqrt{2}}\delta_0,$ we have
\begin{equation*}
\langle S^*S^* \psi_{0,2},\psi_{0,2}\rangle
=
\left\langle \frac{i}{\sqrt{2}} \delta_0,\frac{1}{\sqrt{2}}(\delta_0+i\delta_2 )\right\rangle
=
\displaystyle \frac{i}{2},
\end{equation*}
and since $SS(\psi_{0,2}) = \frac{1}{\sqrt{2}}(\delta_2+i\delta_4)$, we have
\begin{equation*}
    \langle SS\psi_{0,2},\psi_{0,2}\rangle
   =
\left\langle \frac{1}{\sqrt{2}}(\delta_2+i\delta_4),\frac{1}{\sqrt{2}}(\delta_0+i\delta_2)\right\rangle 
=
\frac{1}{2} \langle \delta_2,i\delta_2\rangle
=
-\frac{i}{2}.
\end{equation*}
Hence, the lower bound in Corollary~\autoref{thm:Main_result} obtained from the $R_{\psi_{0,2}}$-matrix is
    \begin{equation*}
    \Delta^2_{\psi_{0,2}}(S)+\Delta^2_{\psi_{0,2}}(S^*)
\geq 
    \left|    \langle S^*S^*(\psi_{0,2}),\psi_{0,2}\rangle - \langle SS(\psi_{0,2}),\psi_{0,2}\rangle\right|
    =
    \left| \frac{i}{2} - \left(-\frac{i}{2}\right)\right|
    =
    1.
    \end{equation*}

\end{enumerate}
\end{example}

\begin{example}
Here we look at the gamma matrices themselves as observables, specifically
for the case $d=5$. Collections of mutually anticommuting gamma matrices also arise naturally in Dirac-type Hamiltonians used to describe relativistic particles and topological band structures. In particular, five such matrices provide a representation suitable for four-band Hamiltonians with five independent coefficients, making this example relevant to the algebraic structure underlying higher-dimensional topological insulators. We find, for real $\vec{\lambda}$, that
\begin{align*}
Q_{\vec{\lambda}}(\Gamma_{1},\dots,\Gamma_{5}) & =\sum\left(\Gamma_{j}-\lambda_{j}I\right)^{2}\\
 & =\sum\left(1+\lambda_{j}^{2}\right)I-2\lambda_{j}\Gamma_{j}.
\end{align*}
We know \cite{cerjan2025multivariable_PS} that this has rotation symmetry, and we now see that the quadratic gap has minimum $4$ that occurs on the unit
sphere. Lets focus on $\vec{\mu}=(0,0,0,0,1)$. If we select our matrices
concretely as
\begin{equation}
\begin{aligned}\Gamma_{1} & =\sigma_{x}\otimes\sigma_{x}, & \Gamma_{2} & =\sigma_{x}\otimes\sigma_{y}, & \Gamma_{3} & =\sigma_{x}\otimes\sigma_{z},\\
\Gamma_{4} & =\sigma_{y}\otimes I_{2}, & \Gamma_{5} & =\sigma_{z}\otimes I_{2}.
\end{aligned}
\end{equation}
then we find 
\begin{equation*}
Q_{\vec{\mu}}(\Gamma_{1},\dots,\Gamma_{5})=\left[\begin{array}{cc}
4 & 0\\
0 & 8
\end{array}\right]\otimes I.
\end{equation*}
As one example, the unit vector 
\begin{equation*}
\mathbf{v}=\left[\begin{array}{c}
1\\
0
\end{array}\right]\otimes\left[\begin{array}{c}
1\\
0
\end{array}\right]
\end{equation*}
will achieve the minimum
\begin{equation*}
\sum\Delta_{\mathbf{v}}^{2}\Gamma_{i}=4.
\end{equation*}
As to the lower-bound in our main theorem, we know from \cite{Cerjan24}  that the Clifford gap is zero on the unit sphere. All the expectations of all the
commutators are zero except 
\begin{equation*}
\left\langle i\left[\Gamma_{1},\Gamma_{2}\right]\right\rangle _{\mathbf{v}}=\left\langle i\left[\Gamma_{3},\Gamma_{4}\right]\right\rangle _{\mathbf{v}}=-2
\end{equation*}
so 
\begin{equation*}
\left\Vert R_{\mathbf{v}}(\Gamma_{1},\dots,\Gamma_{5})\right\Vert =\sqrt{8}.
\end{equation*}
Thus the main theorem leads to a lower-bound of $\sqrt{8}$ that so
this bound is not saturated in this example.
\end{example}

\section{Computational examples}

Here we look at some examples derived with the assistance of either
a computer algebra system or numerical methods. When finding maximally
coherent states or Clifford spectrum, it is usually a bad idea to
work with exact arithmetics. Only when working with $n$-by-$n$ matrices
for $n\ll10$ will hand calculations suffice, and if a computer algebra
system (CAS) is employed, $n\approx50$ can be feasible \cite{schreivogl2024fuzzy}
if the matrices are sparse enough. See especially the matrices in
\cite{schreivogl2024fuzzy} regarding the quantum alphabet to see
how such sized matrices can be dealt with, carefully, using the computer
to track the formulas.

In our first example, we will one observable non-Hermitian and the
one other observable will be Hermitian. To get down to three real
dimensions, from $\mathbb{C}^{2}$ to $\mathbb{C}\times\mathbb{R}$,
we will use the following, simple lemma whose proof we omit.

\begin{lem}
If $A_{j}^{*}=A_{j}$ for $1\leq j\leq d$ and $\mathbb{A}=(A_{1},\dots,A_{d})$
in $\mathcal{A}$ and if $\vec{\lambda}\in\mathbb{C}^{d}$ minimizes
$\mu_{\vec{\lambda}}^{\mathrm{C}}(\mathbb{A})$ then $\lambda_{j}\in\mathbb{R}$.
\end{lem}


\begin{example}
Consider 
\begin{equation*}
A=\left[\begin{array}{cc}
0 & 1\\
1 & 0
\end{array}\right],\ B=\left[\begin{array}{cc}
1 & 0\\
1 & 1
\end{array}\right].
\end{equation*}
Two-dimensional matrices provide the simplest setting for describing quantum two-level systems, while non-Hermitian operators naturally arise in effective descriptions of open systems and systems with asymmetric couplings. This example illustrates how the quadratic pseudospectrum can be used to determine optimal uncertainty bounds even when the operators are not restricted to conventional Hermitian observables. We can calculate the minimum of
$V := \Delta_{\mathbf{v}}^{2}A+\Delta_{\mathbf{v}}^{2}B$ 
as $\mathbf{v}$ varies over the Bloch sphere.  This is routine, with the vectors
\begin{equation*}
\mathbf{v}=\frac{1}{\sqrt{5}}\left[\begin{array}{c}
\sqrt{2}\\
\pm\sqrt{3}
\end{array}\right].
\end{equation*}
being the two states achieve the minimum of $V=\tfrac{1}{5}$.  The calculation via the quadratic pseudospectrum is a bit lengthy so we used a computer algebra system here.  See the auxiliary file 
\texttt{localizer\_state\_one\_non\_normal.nb} \cite{Loring_Uncertainty_code_2026} for details.
Since $A$ is Hermitian, we need only
use $\vec{\lambda}=(x,y+iz)$ in $\mathbb{R}\times\mathbb{C}$. We
find
\begin{equation*}
Q_{(x,y,z)}(A,B)=\left[\begin{array}{cc}
x^{2}+(y-2)y+z^{2}+3 & -2x-y-iz+1\\
-2x-y+iz+1 & x^{2}+(y-2)y+z^{2}+2
\end{array}\right].
\end{equation*}
The smaller of the two eigenvalues of this is
\begin{equation*}
\frac{1}{2}\left(2x^{2}+2y^{2}-4y+2 z^{2}+5-\sqrt{(4x+2y-2)^{2}+4z^{2}+1}\right).
\end{equation*}
The partial derivative of this eigenvalue in $z$ is
\begin{equation*}
\frac{2z\left(\sqrt{(4x+2y-2)^{2}+4z^{2}+1}-1\right)}{\sqrt{(4x+2y-2)^{2}+4z^{2}+1}}
\end{equation*}
so we must have $z=0$ to achieve a minimum. The function
\begin{equation*}
\mu\left(Q_{(x,y,0)}(A,B)\right)=\frac{1}{2}\left(2x^{2}+2y^{2}-4y+5-\sqrt{(4x+2y-2)^{2}+1}\right)
\end{equation*}
is minimized at two points, 
\begin{equation*}
(x,y)=\left(\pm\frac{2\sqrt{6}}{5},1\pm\frac{\sqrt{6}}{5},\right).
\end{equation*}
The eigenvectors that correspond to the smaller eigenvalue at this
$\vec{\lambda}$ are 
above.
\end{example}

Given a topological insulator, a free-fermion model on a lattice with open boundaries, we know that at the fermi level the Clifford gap will be rather large except near the the boundary, where it must close.
Of course here we must work numerically. As the emphasis in this work is the mathematics, we will use the exact models used before in the study of pseudospectrum and not give details here on the models used.

\begin{example}
We look first at a Chern insulator, specifically the model used in
\cite[Sect.~9]{Loring15}. Chern insulators are a paradigmatic example of topological matter, in which a nontrivial bulk band topology gives rise to robust boundary states. By considering the Hamiltonian $H$ together with the two position operators, the corresponding uncertainty relations probe the simultaneous localization of a quantum state in energy and real space, particularly near the topological boundary. Here, we use position matrices $X$ and $Y$, with $-L/2\leq X,Y\leq L/2$
such that $L$ is the length and width of the system in terms of lattice
sites, and set $\kappa=0.02$ (as in \cite[Figs. 9.1-5]{Loring15})
so we are looking at the triple $(H,\kappa X,\kappa Y)$. We now
look at the $\vec{\lambda}=\vec{\lambda}(t)$ along a line from $\vec{\lambda}(0)=(0,0,0)$
to $\vec{\lambda}(1)=(0,L/2,0)$. We compute both gaps, $\mu_{\vec{\lambda}}^{Q}(H,\kappa X,\kappa Y)$
and $\mu_{\vec{\lambda}}^{C}(H,\kappa X,\kappa Y)$ as well as the
the lower bound
\begin{equation}
\sqrt{\mu_{\vec{\lambda}}^{C}(H,\eta X,\eta Y)+\sqrt{\left(\left\langle \left[H,\kappa X\right]\right\rangle _{\mathbf{v}}\right)^{2}+\left(\left\langle \left[H,\kappa Y\right]\right\rangle _{\mathbf{v}}\right)^{2}}}\label{eq:good_lower_bound}
\end{equation}
where $\mathbf{v}$ is selected as an eigenvector for the smallest
eigenvalue of $Q_{\vec{\lambda}}(H,\kappa X,\kappa Y)$. This data
is plotted in Figure~\ref{fig:Chern_30} Here we are pleasantly surprised
to find, with $L=40$, that this lower bound is nearly equal to $\mu_{\vec{\lambda}}^{Q}(H,\kappa X,\kappa Y)$
at near edge of the topological insulator. We also compute, for the
same $\mathbf{v}$,
\begin{equation*}
\sqrt{\Delta_{\mathbf{v}}^{2}H+\Delta_{\mathbf{v}}^{2}\kappa X+\Delta_{\mathbf{v}}^{2}\kappa Y}
\end{equation*}
and compare that to the value of the quadratic pseudospectrum. Again,
it is just near the edge where these are nearly equal. For large systems
the approximations at the edge get more exact, as shown in Figure~\ref{fig:Chern_40}. The near-saturation of this bound at the boundary demonstrates that the pseudospectral formulation can provide particularly sharp constraints on the energy and spatial localization of topological edge states.
\end{example}

\begin{figure}
\includegraphics[scale=0.45]{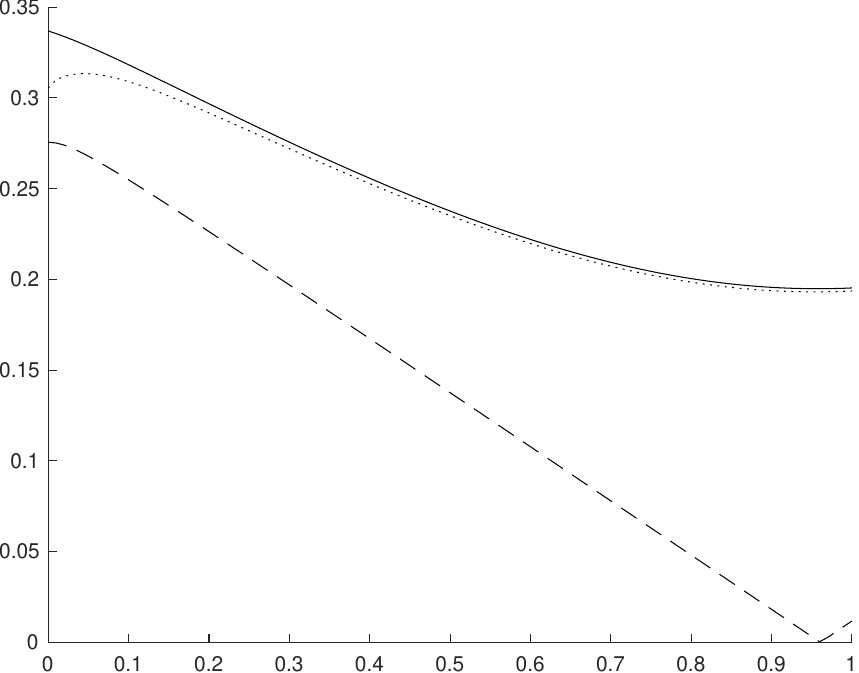}\hspace{0.5cm}\includegraphics[scale=0.45]{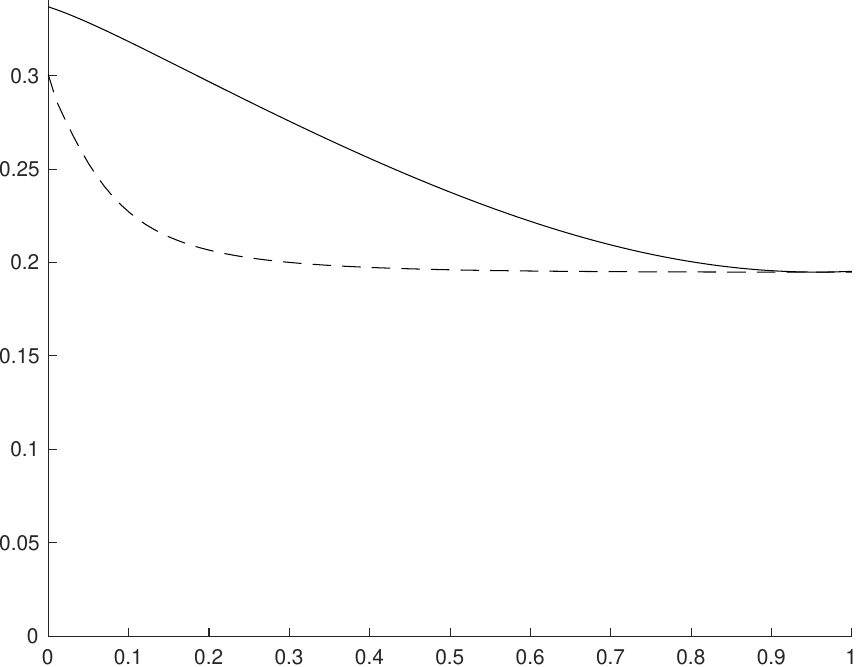}\caption{This figure shows data from a Chern insulator of size $L=30$. The left panel shows the quadratic gap as a solid line, the Clifford gap as a dashed line, and the lower-bound on the quadratic gap as in (\ref{eq:good_lower_bound})
using alternating dashes and dots. The right panel also shows the quadratic gap and compares this to the square-root of the sum of variances
(dashed line). \label{fig:Chern_30}}
\end{figure}

\begin{figure}
\includegraphics[scale=0.45]{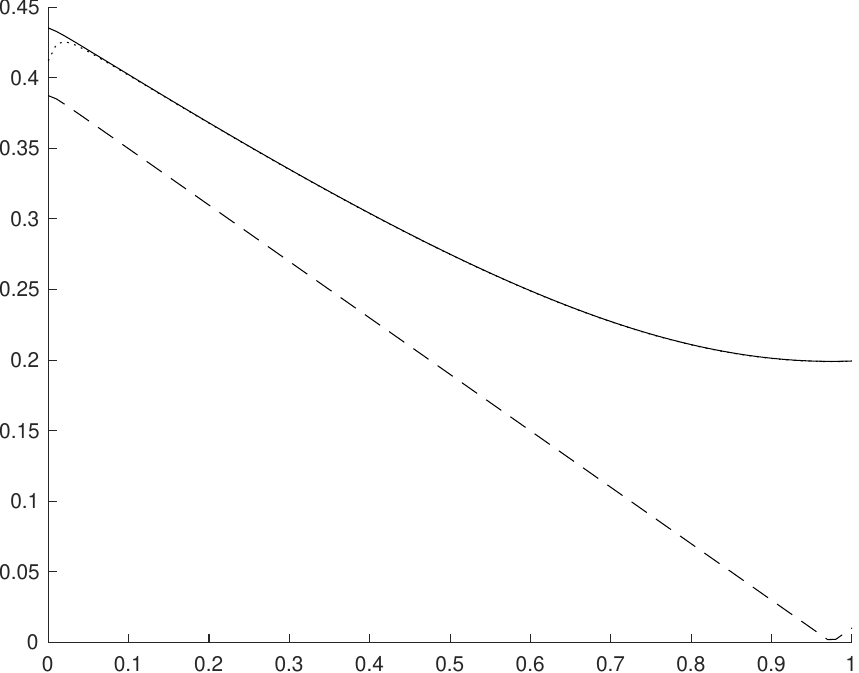}\hspace{0.5cm}\includegraphics[scale=0.45]{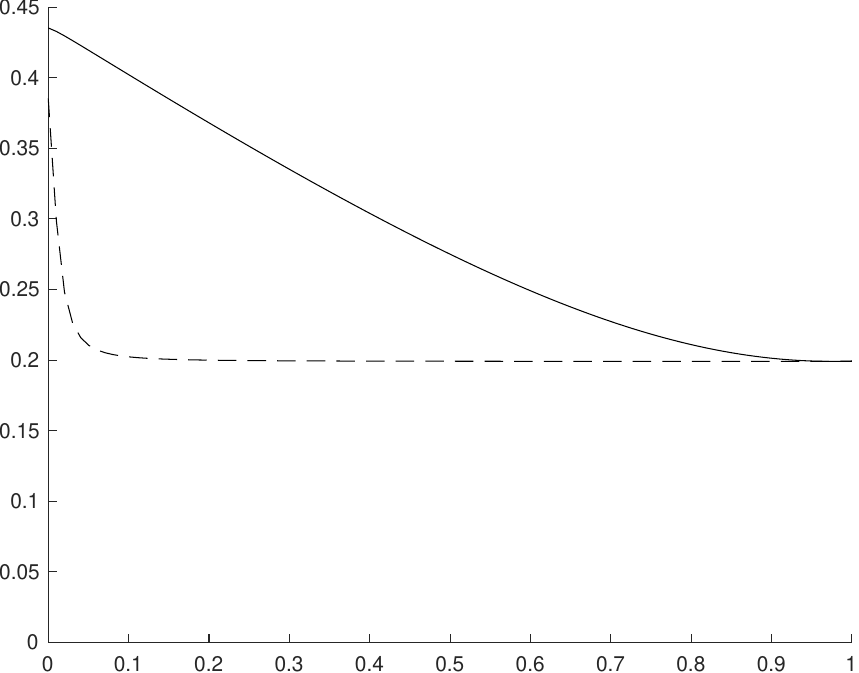}\caption{Same Chern insulator and plots as in Figure~\ref{fig:Chern_30}
but now with $L=40$. The left panel shows the quadratic gap as a
solid line, the Clifford gap as a dashed line, and the lower-bound
on the quadratic gap as in (\ref{eq:good_lower_bound}) using alternating
dashes and dots. The right panel also shows the quadratic gap and
compares this to the square-root of the sum of variances (dashed line).\label{fig:Chern_40}}
\end{figure}

\begin{example}
We repeat the above for a 4D Chern insulator. Recall that 4D Chern insulators extend the topology of 2D Chern insulators to systems characterized by a nontrivial second Chern number. Although four spatial dimensions are not directly accessible in conventional electronic materials, such higher-dimensional topological phenomena can be explored using synthetic dimensions and engineered electronic or photonic systems, making these models relevant to experimentally realizable topological physics. We use the same model
as in \cite[Sec.~7]{Cerjan24}. Thus we have $H$ along with position
operators $X_{1}$ through $X_{4}$. We use $\kappa=0.10$ as was
used in \cite[Sec.~7]{Cerjan24}. We again plot the two psdeudospectral
gaps, as well as
\begin{equation}
\sqrt{\mu_{\vec{\lambda}}^{C}(H,\eta X_{1},\dots,\eta X_{4})+\sqrt{\sum_{j=1}^{4}\left(\left\langle \left[H,\kappa X_{i}\right]\right\rangle _{\mathbf{v}}\right)^{2}}}\label{eq:4D_lower_bound}
\end{equation}
and 
\begin{equation*}
\sqrt{\Delta_{\mathbf{v}}^{2}H+\sum_{j=1}^{4}\Delta_{\mathbf{v}}^{2}\kappa X_{j}}.
\end{equation*}
We do this only for $L=5$ and $L=8$, so what we are seeing might
be an artifact of using such a small system size. Nevertheless, this example further illustrates the applicability of our uncertainty bounds to multiple spatial coordinates, as required for characterizing quantum localization in higher-dimensional topological systems.
\end{example}

\begin{figure}
\includegraphics[scale=0.45]{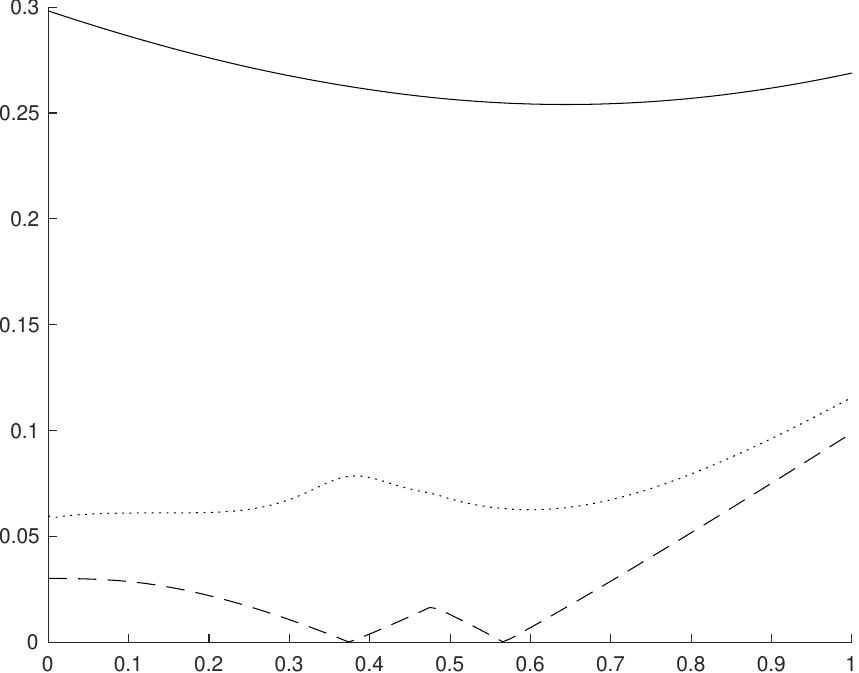}\hspace{0.5cm}\includegraphics[scale=0.45]{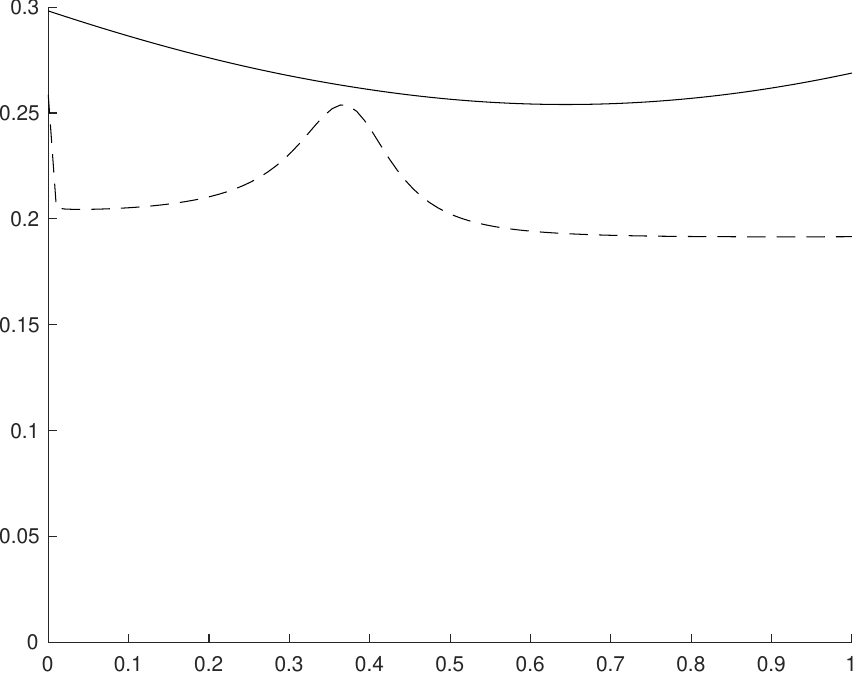}\caption{This figure shows data from a 4D, class A topological insulator of
size $L=5$. The left panel shows the quadratic gap as a solid line,
the Clifford gap as a dashed line, and the lower-bound on the quadratic
gap as in (\ref{eq:4D_lower_bound}) using alternating dashes and
dots. The right panel also shows the quadratic gap and compares this
to the square-root of the sum of variances (dashed line). \label{fig:4D_classA_L5}}
\end{figure}
\begin{figure}
\includegraphics[scale=0.45]{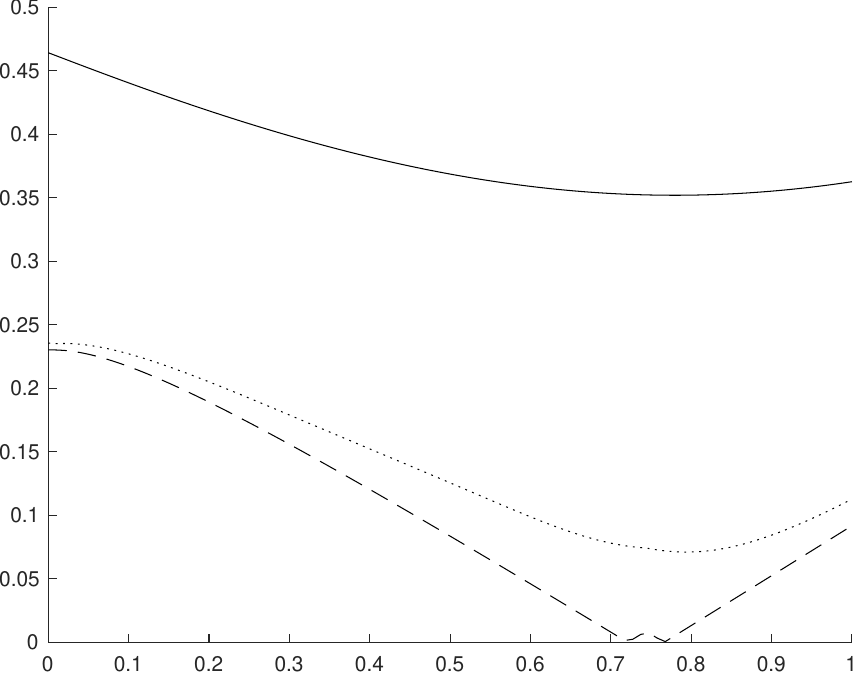}\hspace{0.5cm}\includegraphics[scale=0.45]{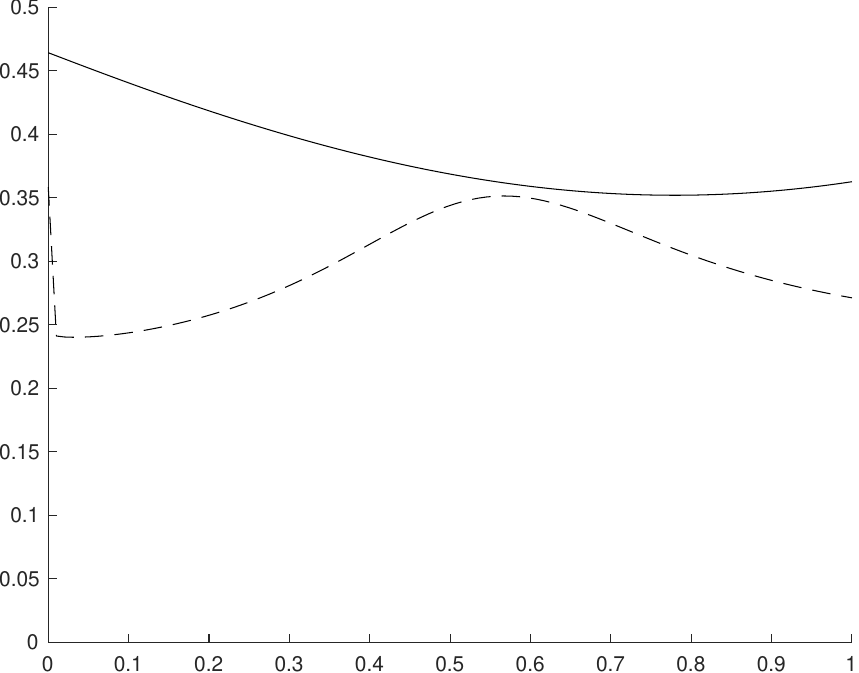}\caption{This is the same as shown in Figure \ref{fig:4D_classA_L5}, but now
$L=8$. \label{fig:4D_classA_L8}}
\end{figure}

\section*{Supplementary Material}
Matlab and Mathematica files that perform some of the calculations are posted at \cite{Loring_Uncertainty_code_2026}.

\section*{Acknowledgments}
T.A.L.\ acknowledges support from the National Science Foundation, Grant No.\ DMS-2349959. Research by
L.I.\  was sponsored by the Army Research Office and
was accomplished under Grant Number W911NF-25-1-
0052. 
A.C.\ acknowledges support from the Laboratory Directed Research and Development program at Sandia National Laboratories.
This work was performed in part at the Center for Integrated Nanotechnologies, an Office of Science User Facility operated for the U.S.\ Department of Energy (DOE) Office of Science.
Sandia National Laboratories is a multimission laboratory managed and operated by National Technology \& Engineering Solutions of Sandia, LLC, a wholly owned subsidiary of Honeywell International, Inc., for the U.S.\ DOE's National Nuclear Security Administration under Contract No.\ DE-NA-0003525.
The views expressed in the article do not necessarily represent the views of the U.S.\ DOE or the United States Government.

\bibliographystyle{plain}

\appendix
\section{Structure of the R matrix}

\begin{lem}\label{lem:R-matrix_Lemma1}
If $a_{jk}$ are arbitrary real scalars for $1\leq j<k\leq d$, let
\begin{equation}
R=\sum_{j<k}a_{jk}(-i\gamma_{j}\gamma_{k})
\end{equation}
where $\gamma_{1},\dots,\gamma_{d}$ form a Clifford representation.
The spectrum (as a set, ignoring multiplicity) of $R$ is independent
of the Clifford representation used.
\end{lem}

\begin{proof}
Since we are ignoring multiplicity in the spectrum, it suffices to
examine the case where $\gamma_{1},\dots,\gamma_{d}$ generate
an irreducible representation of the complex Clifford algebra with
$d$ generators. If $d$ is even, there is only one irreducible representation up to unitary equivalence. If $d$ is odd, both $\{\gamma_{1},\dots,\gamma_{d}\}$ and $\{-\gamma_{1},\dots,-\gamma_{d}\}$ are distinct irreducible representations. The $R$-matrix formed in either representation is the same because the Clifford elements appear in pairs and $\gamma_j\gamma_k = (-\gamma_j)(-\gamma_k)$. Therefore, the spectrum is independent of choice of Clifford representation.
\end{proof}


\begin{lem}
Let $R$ be as in Lemma~\ref{lem:R-matrix_Lemma1}.
Then $R=0$ if and only if all the scalars $a_{jk}$ are zero.
\end{lem}

\begin{proof}
By Lemma~\autoref{lem:R-matrix_Lemma1}, we may assume that the $\gamma_{1},\dots,\gamma_{d}$ generate a copy of the full Clifford algebra. It is well known that a basis for this algebra is given by all products 
\begin{equation*}
\gamma_{j_{1}}\gamma_{j_{2}}\dots \gamma_{j_m}
\end{equation*}
for $0\leq m\leq d$ and $j_{r}<j_{s}$ for $1\leq r<s\leq m$. In particular, the set of matrices $\{(-i\gamma_{j}\gamma_{k})\mid 1\leq j<k\leq d\}$ is linearly independent and the result follows.
\end{proof}


\begin{lem}\label{lem:R-matrix_Lemma3}
 If $R$ is as in (\ref{eq:Def_of_R}), then the spectrum of $R$ is symmetric about $0$. In particular, $-\|R\|$ is always an eigenvalue of $R$.
\end{lem}

\begin{proof}
We first show that there is always a Clifford representation so that all the generators are anti-self-dual.   Recall from
\cite{loring2012factorization}
that the dual operation is usually defined as 
\begin{equation*}
    A^\sharp = -Z A^\top Z
\end{equation*}
where $Z=I\otimes i\sigma_y$ and where
$A^\intercal$ denotes the transpose.  We will proceed by induction on $d$, the number of Clifford matrices.  All three of the Pauli spin matrices are anti-self-dual, so we are set for $d\leq 3$.  For larger $d$ we can start with any gamma matrices $g_1,\dots,g_{d-2}$. A common way to extend this to a list of length $d$ is to set 
\begin{equation*}
\gamma_j = g_j \otimes \sigma_x
\end{equation*}
for $j\leq d-2$ and
\begin{equation*}
\gamma_{d-1} = I \otimes \sigma_y,\quad 
\gamma_{d} = I \otimes \sigma_z
\end{equation*}
which leads to all these being anti-self-dual were the dual is defined by $Z = I \otimes i\sigma_y$.

With this symmetry in hand, we find
\begin{equation*}
\left(\gamma_{j}\gamma_{k}\right)^{\sharp}
=
\gamma_{k}^{\sharp}\gamma_{j}^{\sharp}
=
\gamma_{k}\gamma_{j}
=
-\gamma_{j}\gamma_{k}.
\end{equation*}
These relations imply that $R^{\sharp}=-R$.  Since $R$ is also Hermitian, this implies the symmetry in the spectrum of $R$.  See, for example, \cite{altland1997nonstandard}. 
\end{proof}

Finally, we now prove Lemma~\ref{lem:R_in_special_cases}.  The $n=2$  case is trivial, and the $n=3$ case follows from the fact that the three matrices
\begin{equation*}
    -i\sigma_x \sigma_y = \sigma_z,\ 
    -i\sigma_y \sigma_z = \sigma_x,\ 
    -i\sigma_x \sigma_z = -\sigma_y
\end{equation*}
are again Clifford matrices.  Recall we know
\begin{equation}
\label{eqn:localizer_of_scalars}
\sigma \left(L(\alpha_1,\dots,\alpha_d)\right) =
\left \{ \pm \sqrt{\alpha_1 ^2 + \cdots +\alpha_d } \right  \}.
\end{equation}
For larger $d$ the matrices $-i \gamma_j \gamma_k$ do not satisfy any Clifford relations.  Consider the smaller set of the $d-1$ matrices $g_j = -i \gamma_j \gamma_d$ for $j<d$.  These are Hermitian since for $j\neq k$  we have
\begin{equation*}
g_j^* = i  \gamma_d  \gamma_j = -i \gamma_j   \gamma_d =  g_j.
\end{equation*}
Also
\begin{equation*}
g_j^2 = - \gamma_j \gamma_d  \gamma_j \gamma_d=  \gamma_j^2 \gamma_d^2 = I
\end{equation*}
and
\begin{equation*}
g_j g_k = - \gamma_j \gamma_d \gamma_k \gamma_d 
=  \gamma_j \gamma_k 
= - \gamma_j \gamma_k
=- g_k g_j 
\end{equation*}
so these do form a Clifford representation.  So again in this special case we can use (\ref{eqn:localizer_of_scalars}).

\end{document}